\documentclass{IEEEtran}
\usepackage{booktabs}
\usepackage[T1]{fontenc} 
\usepackage{rotating}
\usepackage{amsmath}
\usepackage{epsfig,endnotes}
\usepackage{hyperref}

\usepackage{mathtools}
\usepackage{url}
\usepackage[noend]{algpseudocode}
\usepackage{tabularx,colortbl}
\usepackage{multirow}
\usepackage[Symbolsmallscale]{upgreek}
\usepackage{algorithm}
\usepackage{algpseudocode}
\usepackage{mathtools}
\usepackage{url}
\usepackage{authblk}
\usepackage{tablefootnote}
\usepackage{siunitx}
\usepackage{textcomp}
\usepackage{gensymb}
\usepackage{listings}
\usepackage{romannum}
\usepackage{multicol}
\usepackage{subcaption}
\usepackage{amssymb}
\usepackage{amsthm}
\usepackage[numbers,sort&compress]{natbib}

\newtheorem{theorem}{Theorem}

\usepackage{color} 
\definecolor{mygreen}{RGB}{28,172,0} 
\definecolor{mylilas}{RGB}{170,55,241}

\title{The Adversarial Implications of Variable-Time Inference}

\author{Dudi Biton$^{1}$, Aditi Misra$^{2}$, Efrat Levy$^{1}$, Jaidip Kotak$^{1}$, Ron Bitton$^{1}$, Roei Schuster$^{4}$, Nicolas Papernot$^{2,3}$,  Yuval Elovici$^{1}$, Ben Nassi$^{1,5}$\\\

$^{1}$Ben-Gurion University of the Negev, $^{2}$University of Toronto, $^{3}$Vector Institute, $^{4}$Wild Moose, $^{5}$Cornell Tech
\\ 
\{bitondud, elevy, jaidip, ronbit, nassib\}@post.bgu.ac.il,
\{aditi.misra, nicolas.papernot\}@mail.utoronto.ca,
roei@wildmoose.ai, elovici@bgu.ac.il, bn267@cornell.edu}

\begin{document}
\maketitle
\thispagestyle{empty}

\begin{abstract}

Machine learning (ML) models are known to be vulnerable to a number of attacks that target the integrity of their predictions or the privacy of their training data. To carry out these attacks, a black-box adversary must typically possess the ability to query the model and observe its outputs (e.g., labels). In this work, we demonstrate, for the first time, the ability to enhance such decision-based attacks.
To accomplish this, we present an approach that exploits a novel side channel in which the adversary simply measures the execution time of the algorithm used to post-process the predictions of the ML model under attack.
The leakage of inference-state elements into algorithmic timing side channels has never been studied before, and we have found that it can contain rich information that facilitates superior timing attacks that significantly outperform attacks based solely on label outputs. In a case study, we investigate leakage from the non-maximum suppression (NMS) algorithm, which plays a crucial role in the operation of object detectors. In our examination of the timing side-channel vulnerabilities associated with this algorithm, we identified the potential to enhance decision-based attacks. We demonstrate attacks against the YOLOv3 detector, leveraging the timing leakage to successfully evade object detection using adversarial examples, and perform dataset inference. Our experiments show that our adversarial examples exhibit superior perturbation quality compared to a decision-based attack. In addition, we present a new threat model in which dataset inference based solely on timing leakage is performed. To address the timing leakage vulnerability inherent in the NMS algorithm, we explore the potential and limitations of implementing constant-time inference passes as a mitigation strategy.

\end{abstract}

\section{Introduction}

Known threats against machine learning (ML) such as adversarial examples~\cite{wang2019security, biggio2018wild}, dataset/membership inference~\cite{shokri2017membership}, and model stealing~\cite{tramer2016stealing} typically assume an adversary who has at least black-box-querying access to the victim model. 
Such adversaries are inherently constrained due to the limited information the model returns for a given query. In practice, prediction APIs often provide \textit{decision-only} outputs (also called ``label-only'' outputs). 
This is generally considered the minimum amount of information available to querying attackers.

In addition, querying adversaries almost always have the ability to measure inference runtime. 
Previously, this was not thought to be a major advantage, perhaps because the neural inference is often thought of as a series of floating-point arithmetic operations with roughly consistent runtime across inputs.
However, the runtimes of many modern architectures can vary widely, especially if they include hybrid neural and non-neural components. 
For example, multi-exit networks~\cite{teerapittayanon2016branchynet, huang2017multi} optimize runtime by terminating confident predictions early; generative language models use sampling and search algorithms to produce full sentences from next-token predictions~\cite{holtzman2019curious}; and object detectors use non-maximum suppression (NMS) to refine their detection~\cite{redmon2017yolo9000,girshick2015fast,he2017mask,liu2016ssd}.

Knowledge of the timing (i.e., runtime) can benefit attackers by providing them with information on internal computation related to inference.
It is widely known that such information, particularly prediction-confidence scores, can be leveraged to mount attacks that are more potent than label-only attacks.
In decision-only attacks such as adversarial examples and dataset/membership inference attacks, many queries are required to estimate gradients or confidence around every point, and such attacks often produce inferior results~\cite{choquette2021label, chen2020hopskipjumpattack}. 
Thus, if timing leakage reveals meaningful information about inference computation, we expect attackers to use it, along with the prediction provided by the neural network.

In this paper, we demonstrate the significant advantages provided by the ability to measure inference runtime, specifically in the context of object detectors. This ability not only benefits existing adversaries but also opens up new avenues for previously unexplored attacks targeted at object detection systems. By leveraging runtime data, adversaries can improve decision-based attacks and even launch new time-based attacks.

To demonstrate the practicality of our attacks, as a case study, we consider the two-phase inference procedure of object detectors. 
They work as follows: first, a learned/neural component outputs a large set of image bounding boxes that are predicted to contain an object; second, because the previous step's output is likely to contain multiple bounding boxes per object, it serves as input to the non-maximum suppression (NMS) algorithm which tries to filter out overlapping bounding boxes. 
We make two key observations regarding this inference procedure: (1) highly confident predictions of the neural component are more likely to be detected, and (2) the NMS runtime is highly affected by the number of bounding boxes in its input. 
As a corollary, the \emph{NMS runtime closely corresponds to the raw confidence values output by the neural component}.
Leveraging this close correspondence, we mount two attacks, the first of which finds adversarial examples, and the second of which leaks private information through dataset inference.

Our contributions can be summarized as follows: 
\begin{itemize}
    \item We are the first to study algorithmic runtime leakage in hybrid neural/non-neural inference procedures.
    Focusing on the case of the NMS algorithm, we quantify the correspondence between timing and internal inference elements, showing that adversaries can learn accurate information about confidence and bounding box input to the NMS algorithm and that the signal-to-noise of this leakage can be amplified arbitrarily and is easily observed in a realistic deployment, such as remotely over an HTTP connection.
    \item We demonstrate that an attacker leveraging this timing advantage can successfully find adversarial examples and perform dataset inference. We show that by leveraging the timing leak, our adversarial examples yield better results at the $L_2$ norm than traditional decision-based attacks. We also perform dataset inference based solely on the timing leak, which is a novel threat model. We provide the implementation for creating our adversarial examples in an open-source repository\footnote{https://github.com/dudi709/Timing-Based-Attack\label{fn:evasion-github}}.
\end{itemize}

The remainder of the paper is structured as follows: in Section \ref{sec:leakage}, we characterize and quantify the NMS algorithm's timing leakage.
In Sections \ref{sec:evasion} and \ref{sec:membership-inference}, we present two new attacks that exploit this leakage to perform an evasion attack and dataset inference, and evaluate their performance. In Section \ref{sec:related-work}, we review related work. We discuss countermeasures in Section \ref{sec:countermeasures} and limitations in Section \ref{sec:limitations}. In Section \ref{sec:discussion}, we discuss the findings of the study and our plans for future work.
\section{Background}
\label{sec:background}

\begin{figure*}
	\centering
	\includegraphics[width=0.85\linewidth]{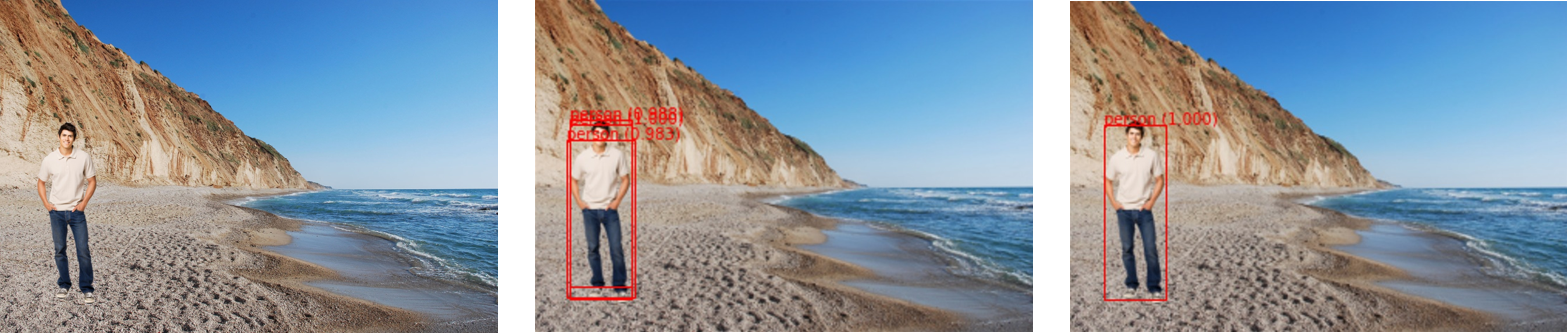}
	\caption{Stages of object detection: original picture (left), the three boundaries (1.0, 0.98, 0.98) annotated by the neural network (middle), and suppression to a single detection by the NMS algorithm (right).}
	\label{fig:object-detection}
\end{figure*}

\paragraph{Object detection}
Object detectors are algorithms that identify objects of a certain class (e.g., pedestrians, road signs, cars) in images. There are two main types of object detectors: those that output the class label along with a confidence score, and those that only output the class label.
During inference, most object detectors perform a neural forward pass, which is followed by the implementation of the NMS algorithm. This process is depicted in Fig. \ref{fig:object-detection}, which presents the original image, the bounding boxes around the detected person (the output of the neural network before applying the NMS algorithm), and the final output (after the NMS algorithm has been applied).
The neural network architecture works by dividing the image into a collection of frames/bounding boxes (potentially overlapping) and assigns a separate score or confidence level to each bounding box indicating the detector's confidence that the object is present in that frame. These scores are used to filter the frames and enable the selection of just those frames that cross a certain threshold and are likely to contain an object. However, this process often results in multiple bounding boxes that represent the same object.
To refine the output further, the NMS algorithm is applied with the purpose of producing a single high-quality frame per object. The algorithm operates by comparing the overlapping bounding boxes and filtering out those with lower scores, ensuring that the optimal bounding box remains for each object.

\begin{algorithm}[h]
\caption{Non-Maximum Suppression}
\label{alg:nms}
\begin{algorithmic}[1]
\State \textbf{Input}: $B = \{B_1,..,B_n\}$, $S = \{S_1,..,S_n\}$, $N_t$ 
\State {B is the list of initial detection boxes}
\State {S contains corresponding detection scores}
\State {$N_t$ denotes the NMS thresholds}

\State \textbf{Initialization}: 
\State $D \leftarrow \{\}$ 
\While {$B^c \neq empty$}
	\State $m \leftarrow argmax(S)$ 
	\State $M \leftarrow b_m$
	\State $D \leftarrow D \cup M$ 
	\State $B \leftarrow B-M$ 
	\For{$b_i \in B$}
		\If{$IoU(M, b_i) \geq N_t$}
			\State $B \leftarrow B-b_i$ 
			\State $S \leftarrow S-s_i$ 
		\EndIf
\EndFor	
\EndWhile
\State \textbf{Output}: $D$, $S$
\end{algorithmic}
\end{algorithm}

\paragraph{Non-maximum suppression (NMS) algorithm}
The NMS algorithm operates based on the principle that two highly overlapping bounding boxes, as indicated by an intersection-over-union (IoU) value that crosses a certain threshold, likely refer to the same object. 
By far the most commonly used approach to implement the NMS algorithm (although improvements~\cite{rothe2014non} and even neural approaches~\cite{hosang2017learning} have been suggested) is a greedy algorithm (see Algorithm~\ref{alg:nms}) which iteratively finds highly overlapping bounding box pairs and filters out the one with the lower score.
More specifically, given a set of input bounding boxes, the NMS algorithm iteratively finds the bounding box with the highest score in the set, removes all other bounding boxes that overlap with it from the input set, and moves the one with the highest score from the input set to the output set. It stops when there are no more bounding boxes in the input set.
 
This greedy NMS algorithm (referred to simply as NMS in the remainder of the paper) finds and filters out all bounding boxes of the same object from the input set in every iteration of the outer loop, making its runtime complexity (in terms of the number of bounding box comparisons) $\Theta(o\cdot B)$, where $B$ is the number of bounding boxes input and $o$ is the number of detected objects (usually $o<<B$).
Bounding box comparisons on their own do not usually run in constant time due to the variation in their size. This can result in variable-time inference, despite the tight bounds on the framework-comparison-number. As a consequence, we can expect differences in the actual runtime among inputs.
As discussed later in the paper, we find that $\Theta(o\cdot B)$ is a good proxy in practice.
\section{Profiling NMS Timing Leakage}
\label{sec:leakage}

The characterization of the NMS algorithm's runtime complexity provided in Section~\ref{sec:background} supports the fact that the number of bounding boxes dramatically affects its runtime. 
In this section, we empirically profile this effect, with the aim of quantifying the degree to which potentially sensitive internal inference-state components, such as the number of bounding boxes or confidence scores, leak into runtime measurements which can be noisy.

\paragraph{Experimental setup}
In the set of experiments that follow, we used the COCO-MS dataset \cite{lin2014microsoft} and the official implementation of YOLOv3 \cite{redmon2016you}, which includes the greedy NMS algorithm described in Section~\ref{sec:background}. We downloaded a pretrained version of YOLOv3, and the pretrained weights were downloaded from the YOLO repository.\footnote{\label{keras} \url{https://github.com/experiencor/keras-yolo3}}  
The rest of YOLO's code was implemented using the scripts provided in \cite{brownlee2019deep}.
We performed inference on an RTX 2080 Ti machine with SIX CPU cores and 32 GB RAM and sampled the execution time of the neural network, NMS algorithm, and YOLO using Python's \texttt{time} module.

\begin{figure}
\setlength{\columnsep}{0.5em}
\begin{multicols}{2}
  \begin{subfigure}[b]{\columnwidth}
    \includegraphics[width=\linewidth]{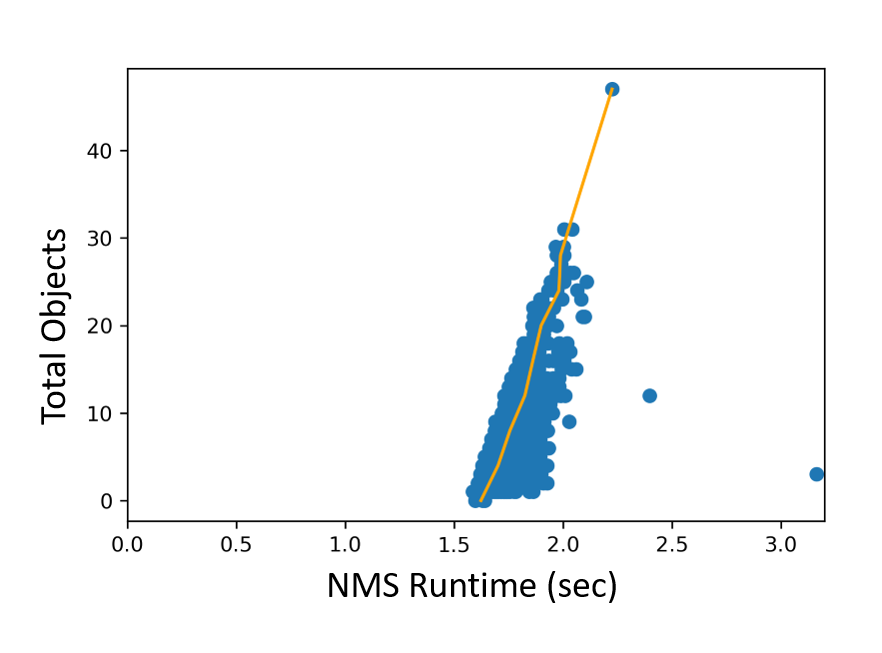}
    \caption{NMS execution time vs. number of objects.}
    \label{fig:nms_objects}
  \end{subfigure}
  \par 
  \begin{subfigure}[b]{\columnwidth}
    \includegraphics[width=\linewidth]{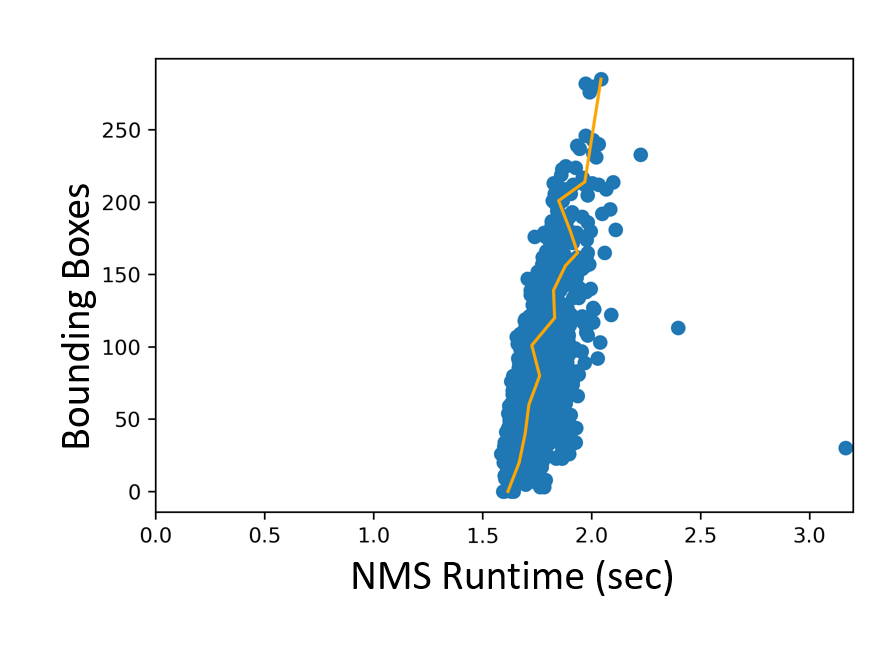}
    \caption{NMS execution time vs. number of bounding boxes.}
    \label{fig:nms_bb}
  \end{subfigure}
  \par 
\end{multicols}
\begin{multicols}{2}
  \begin{subfigure}[b]{\columnwidth}
    \includegraphics[width=\linewidth]{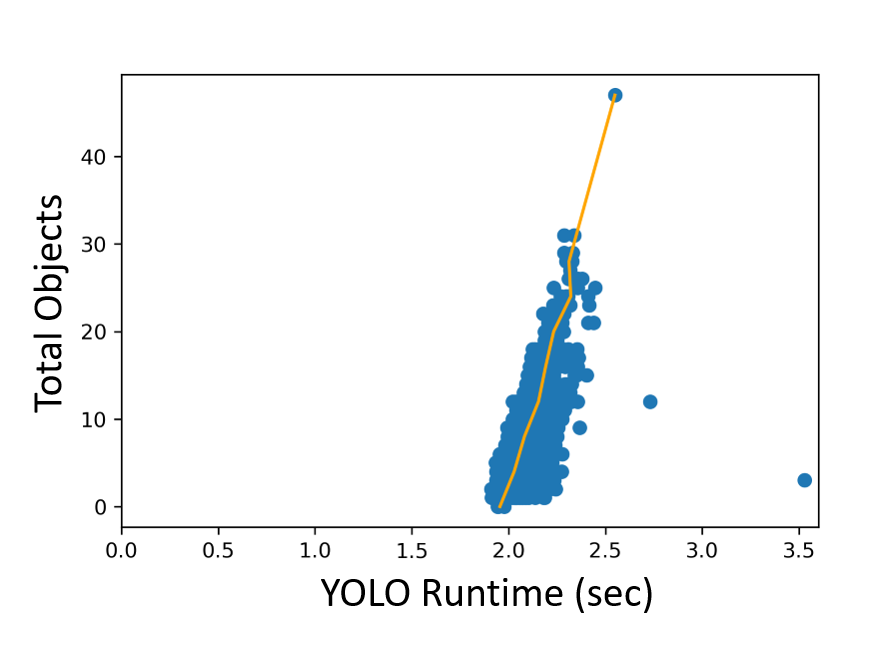}
    \caption{YOLO execution time vs. number of objects.}
    \label{fig:yolo_objects}
  \end{subfigure}
  \par 
  \begin{subfigure}[b]{\columnwidth}
    \includegraphics[width=\linewidth]{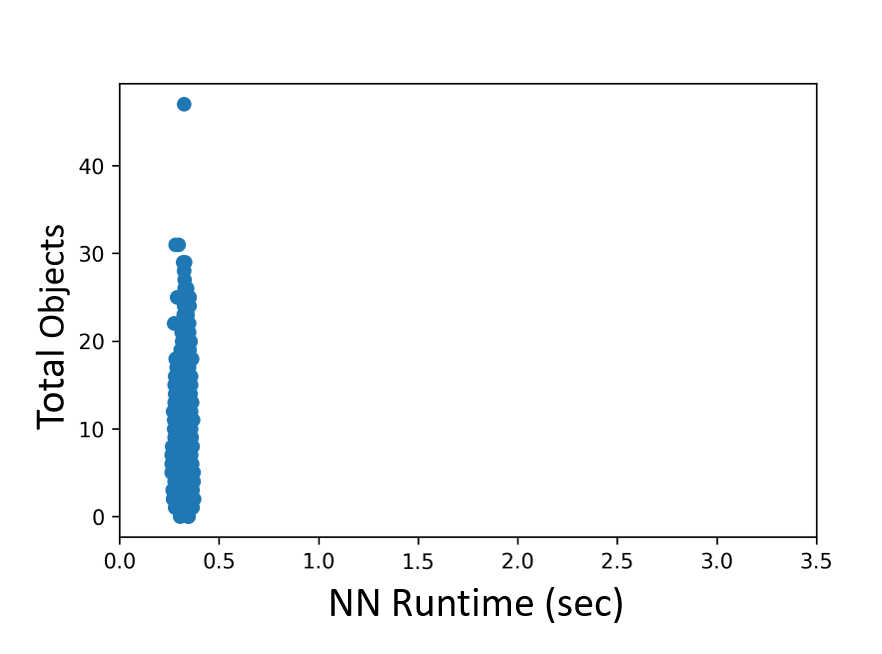}
    \caption{Neural network execution time vs. number of objects.}
    \label{fig:nn_objects}
  \end{subfigure}
\end{multicols}
\caption{The runtime of various YOLOv3 components as a function of the number of objects that appear in the images. Subfigures a and b depict the NMS algorithm's execution time, while subfigures c and d show YOLO's total inference time and the neural network's execution time, respectively.} 
\label{fig:analysis}
\end{figure}

\paragraph{Correspondence between the number of bounding boxes for a varied number of objects that appear in an image and the NMS and YOLO runtime}
Given an image, in which the number of objects varies, we expect the NMS runtime to increase correlatively as the number of objects in the image increases (see Section~\ref{sec:background}). 
To verify this, we randomly selected 1,500 images from the dataset, each of which contain up to 50 objects, and sent them to YOLO. 
Then we sampled the execution time of the neural network and the NMS algorithm and the end-to-end inference time of YOLO.
The results are presented in Fig. \ref{fig:analysis}, where the orange line indicates the average behavior. Figs. \ref{fig:nms_objects} and \ref{fig:nms_bb} show that the execution time of the NMS algorithm increases as a function of both the number of objects and the number of bounding boxes.
Moreover, there is a linear connection between the number of objects in an image and the execution time of the NMS algorithm. 
An image with more objects creates more bounding boxes, which results in more merging operations required by the algorithm.
Fig. \ref{fig:nn_objects} shows that the execution time of the neural network is fixed and is not affected by the number of objects that appear in an image, while Fig. \ref{fig:yolo_objects} shows that the linear correlation between the execution time and the number of objects affects YOLO's total runtime due to the fact that YOLO's execution time (excluding the execution time of the NMS algorithm) is fixed. 
In addition, Figs. \ref{fig:yolo_objects} and \ref{fig:nms_objects} show that the difference in the execution times of YOLO and the NMS algorithm on two images, in which the first image contains a small number of objects (e.g., three) and the second image contains a large number of objects (e.g., 30) is much more pronounced than the difference in the execution times of YOLO and the NMS algorithm on two images that contain around the same number of objects (e.g., five and six). 

 
\paragraph{Correspondence between the number of bounding boxes for a single object and the NMS runtime}
Given a single object, we expect the NMS runtime to be highly correlated with the number of bounding boxes (see Section~\ref{sec:background}). 
To verify this, we randomly selected 1,500 images from the dataset, each of which contain up to 30 objects.
The images were handled differently, depending on whether they contained just one object or more than one object; we extracted the frames that only contain one object; this is done automatically by iterating the images and executing YOLOv3 with a threshold of 0.6. 
Images containing a single object remained as is. For each image that contains more than one object, we cropped each object by its detected bounding box and expanded it by 20 pixels in each direction, and used it as an image on its own. 
We resized the height and width of images to the nearest multiple of 32 due to YOLO's input size requirements, and we ran an inference test on them.
The results are presented in Fig. \ref{fig:localnoamp}, verifying a high Spearman correlation of 0.83.


\begin{figure}
\setlength{\columnsep}{0.5em} 
\begin{multicols}{2}
  \begin{subfigure}[b]{\columnwidth}
    \includegraphics[width=\linewidth]{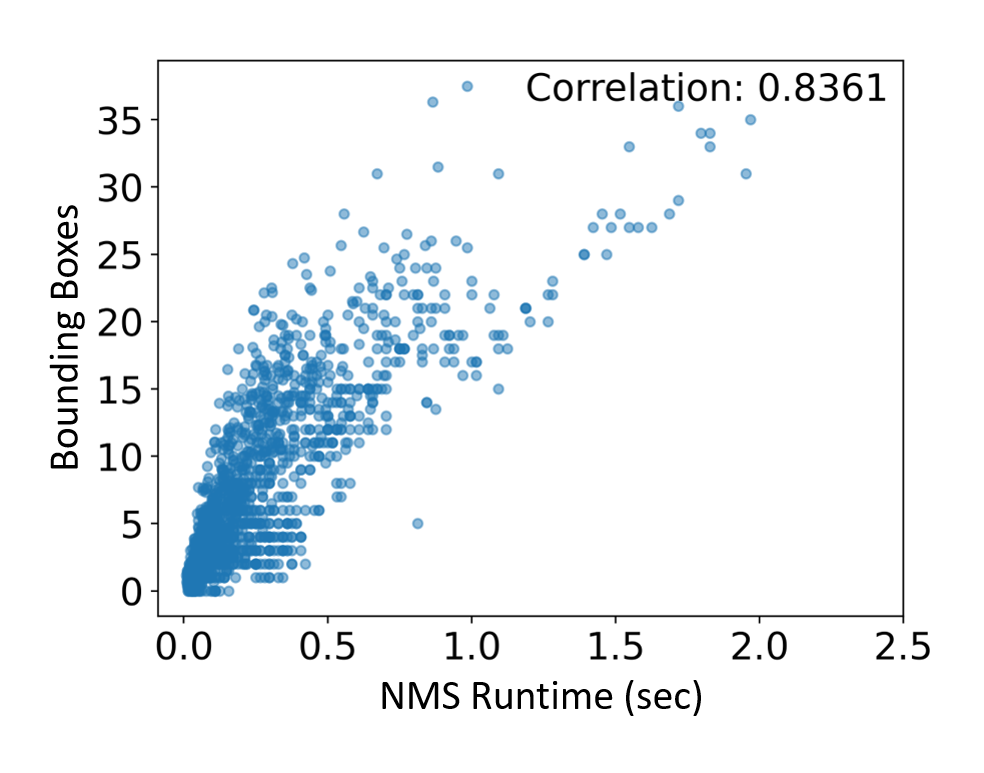}
    \caption{No amplification.}
    \label{fig:localnoamp}
  \end{subfigure}
  \par 
  \begin{subfigure}[b]{\columnwidth}
    \includegraphics[width=\linewidth]{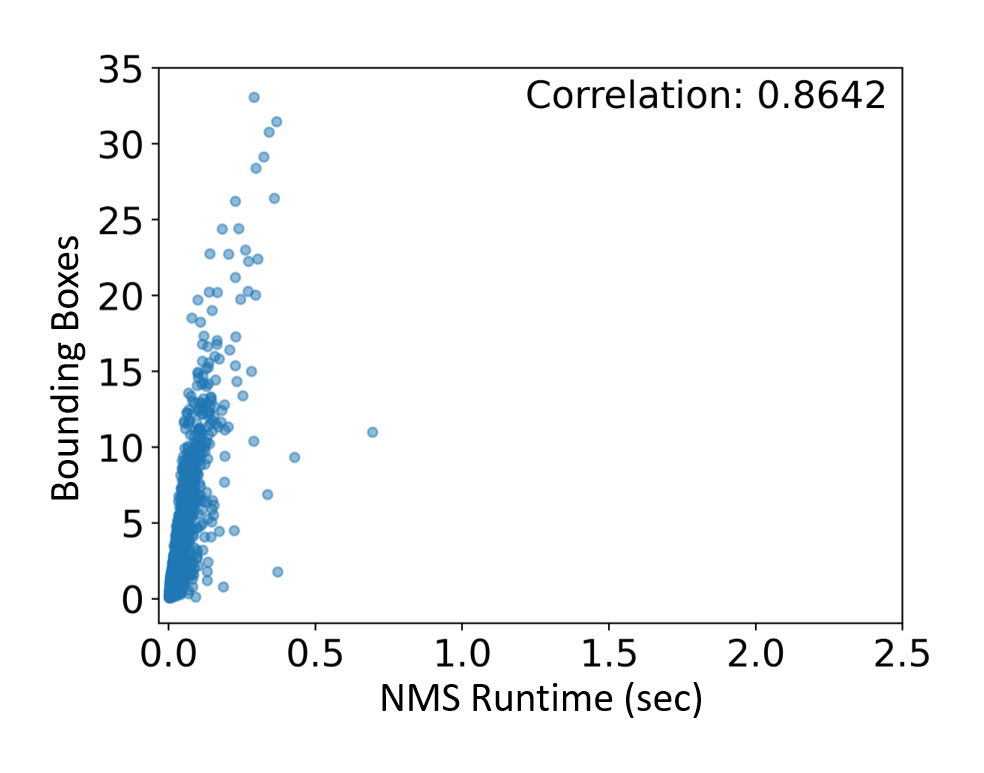}
    \caption{3x3 amplification.}
    \label{fig:local33}
  \end{subfigure}
  \par 
\end{multicols}
\begin{multicols}{2}
  \begin{subfigure}[b]{\columnwidth}
    \includegraphics[width=\linewidth]{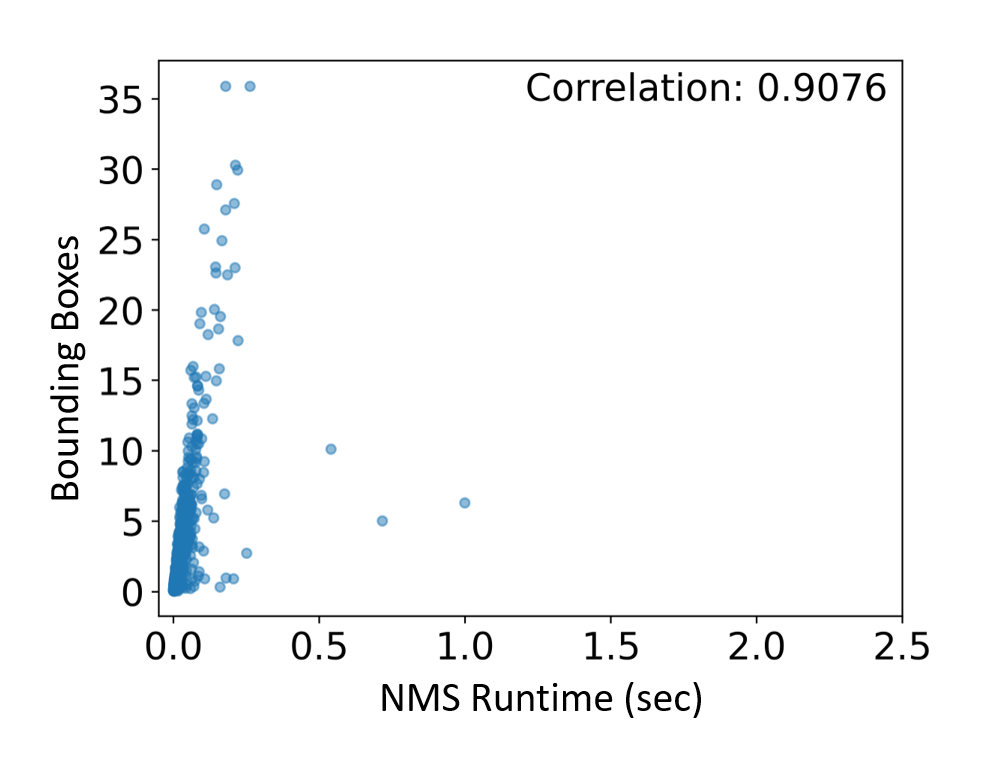}
    \caption{7x7 amplification.}
    \label{fig:local99}
  \end{subfigure}
  \par 
\end{multicols}
\caption{NMS runtime vs. the number of bounding boxes (local querying).}
\label{fig:bbs_nms_local}
\end{figure}

\paragraph{Amplifying leakage signal-to-noise}
Based on the findings of the experiments described above, we now examine whether we can conclude that if $\Theta(o\cdot B)$ is an approximation of the runtime (see Section~\ref{sec:background}), then we can amplify the effect of the runtime's correlation to the number of bounding boxes by tiling the same object image multiple times to form a new image. 
For example, if we concatenate an image containing a single object to itself five times, we expect a five-fold increase in the runtime cost incurred by the NMS algorithm for each bounding box.
This might be useful for an attacker aiming to extract fine-grained information regarding the number of bounding boxes. 
To quantify this effect, Figs. \ref{fig:local33} and \ref{fig:local99} show the runtime vs $B$ graphs of the above objects (those mentioned in the previous paragraph), tiled in a 3x3 and 7x7 pattern, respectively, and confirm higher correlations between the number of bounding boxes and the runtime (0.86 and 0.9, respectively).


\begin{figure}
\setlength{\columnsep}{0.5em} 
\begin{multicols}{2}
  \begin{subfigure}[b]{\columnwidth}
    \includegraphics[width=\linewidth]{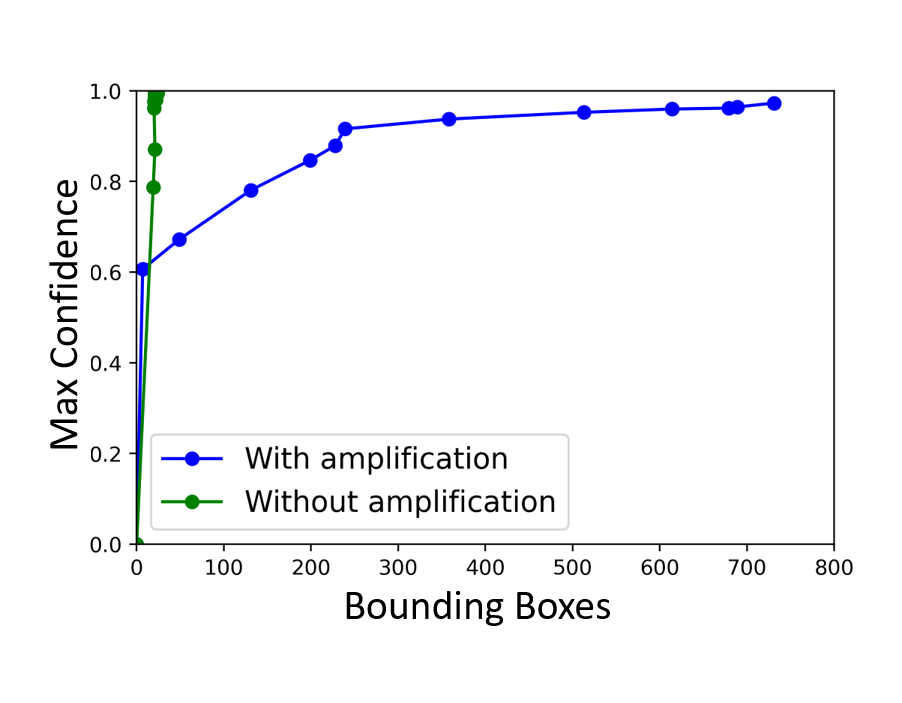}
  \end{subfigure}
  \par 
  \begin{subfigure}[b]{\columnwidth}
    \includegraphics[width=\linewidth]{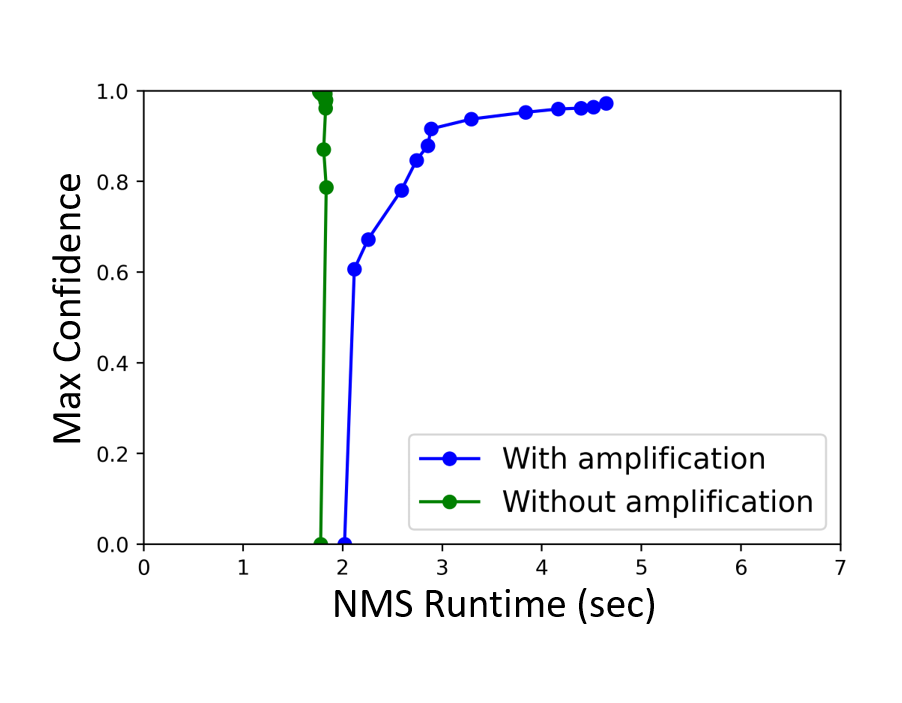}
  \end{subfigure}
  \par 
\end{multicols}
\begin{multicols}{2}
  \begin{subfigure}[b]{\columnwidth}
    \includegraphics[width=\linewidth]{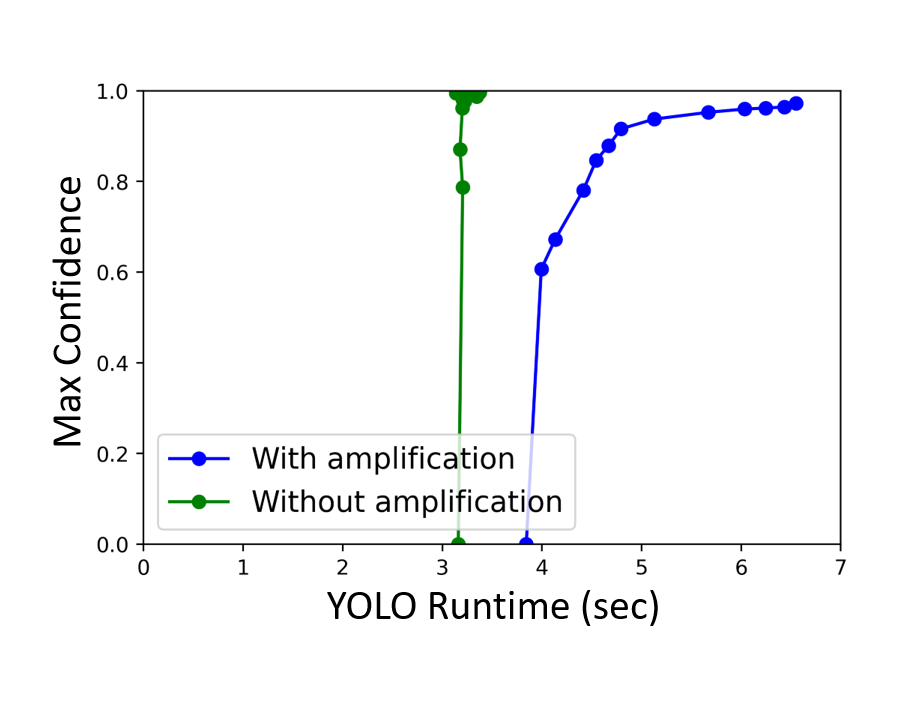}
  \end{subfigure}
  \par 
\end{multicols}
\caption{Confidence level vs. runtime of NMS and YOLO. YOLO was executed with a detection threshold of 0.6.}
\label{fig:amplification}
\end{figure}

\paragraph{Correspondence between leakage and model confidence}
We will see that the number of bounding boxes is closely related to the detector's confidence level. Higher confidence levels across bounding boxes that contain the same object will mean that more bounding boxes cross the first filtering threshold (see Section~\ref{sec:background}) and are input to the NMS algorithm. 
Due to the close relationship between the number of bounding boxes and runtime, we therefore expect the confidence scores to be tied to the runtime as well. 
To examine this, we ran YOLOv3 on the dataset and randomly chose 12 images from the dataset with varying confidence scores (0-1.0). 
Each of the 12 images was amplified (each image was concatenated to itself 140 times).
We ran YOLOv3 on the 12 original images and the amplified versions of the 12 images, and for each of the 24 images, we determined the correspondence between the runtime and the number of bounding boxes input to the NMS algorithm, as well as the confidence level of the bounding boxes output by the algorithm, with and without the tiling amplification procedure. 
Fig. \ref{fig:amplification} shows that amplification shows a strong relationship between the confidence scores and the runtime of both the NMS algorithm and YOLO.

\paragraph{Inferring NMS runtime from total inference runtime}
Attackers who can only query a detector in a black-box fashion would not be able to accurately measure the time of NMS as described in this section. 
Instead, they would only be able to measure the runtime of the entire inference process.
However, we note that this runtime is mainly comprised of (1) the neural network runtime, and (2) the NMS runtime. 
In a fully convolutional neural network like YOLOv3, we expect the first component to be linear with the input image size. 
Furthermore, when no objects are detected, the NMS runtime is negligible. 
Therefore, our querying attacker can use the following method to accurately predict the runtime for a given image size. For example, an attacker can (1) input all-black images of various sizes, for which we expect $B=0$, (2) measure their end-to-end runtime, and fit a linear model that predicts the neural component's runtime for an image of a given size. 
To demonstrate that example, we prepared the images by taking a totally black 416x416 pixel image $b$ and generating 78 new images as follows: for each $k\in\{1, 2,..., 39\}$, we generate two new images by resizing $b$ to both horizontal and vertical rectangles whose sizes are 416x[416+$(k*32)$] and [416+$(k*32)$]x416 pixels, respectively. 
Fig. \ref{fig:nruntimelocal} shows the resulting measurements of the images. 
Then, given a total-inference runtime, our attacker can simply subtract the predicted neural component's runtime from the total runtime to obtain an estimate of the NMS component's runtime. 
Fig. \ref{fig:estvsreallocal} shows the estimated runtime vs. real runtime for the local measurements obtained in our example.

\begin{figure}
\setlength{\columnsep}{0.5em} 
\begin{multicols}{2}
  \begin{subfigure}[b]{\columnwidth}
    \includegraphics[width=\linewidth]{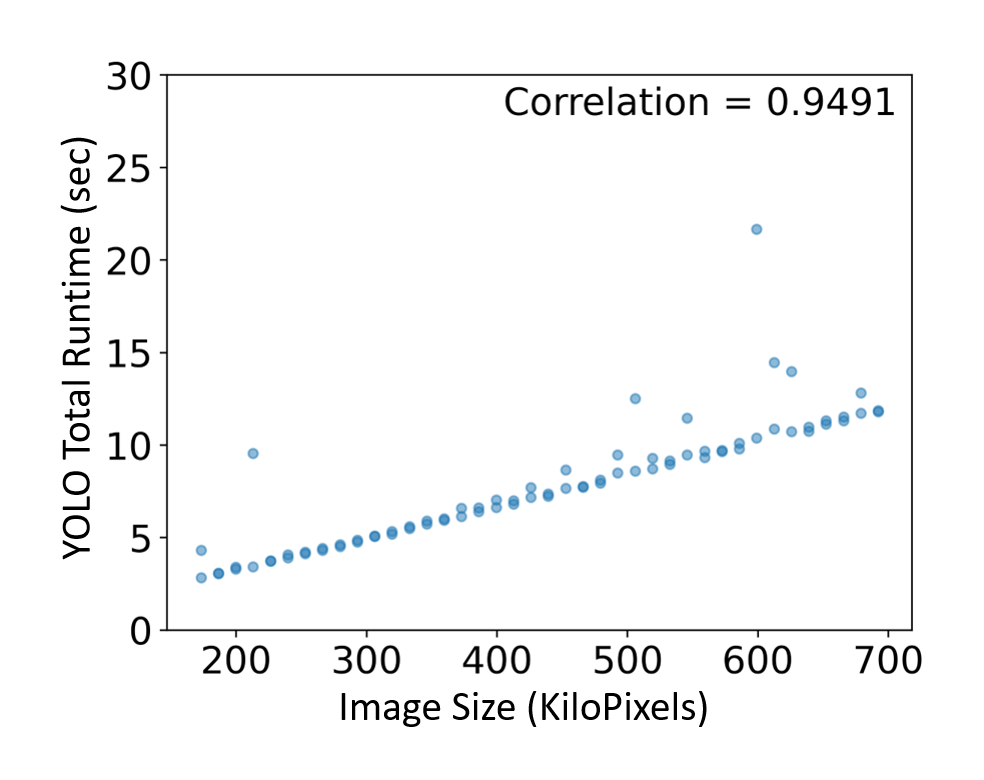}
    \caption{All-black image size vs. runtime (local).}
    \label{fig:nruntimelocal}
  \end{subfigure}
  \par 
  \begin{subfigure}[b]{\columnwidth}
    \includegraphics[width=\linewidth]{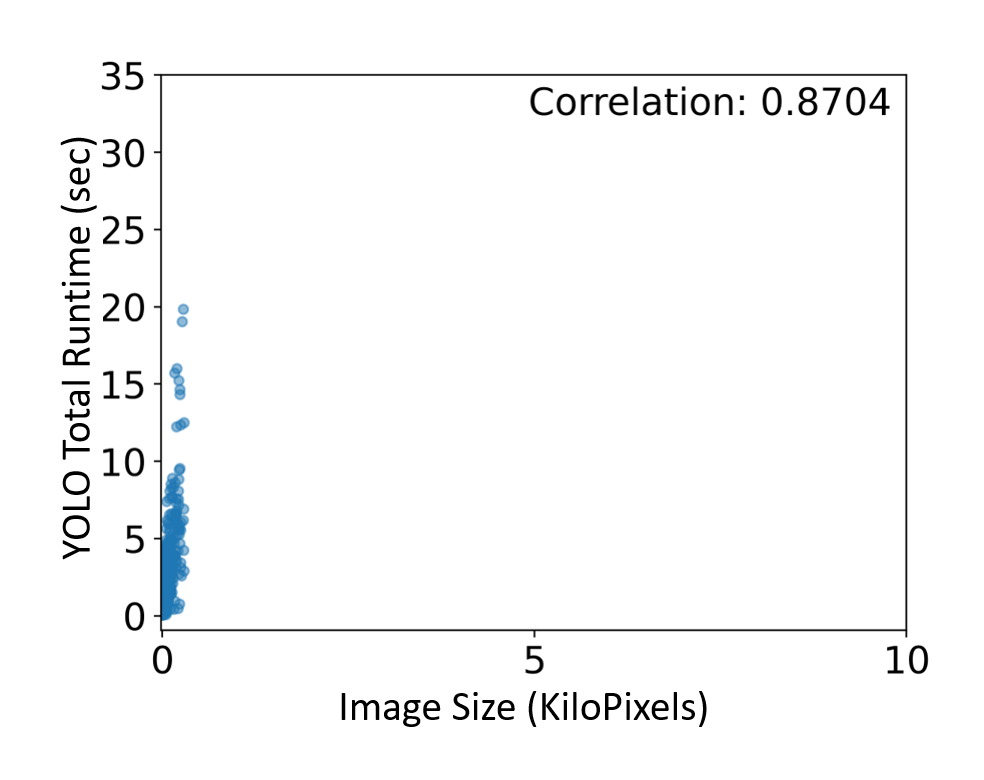}
    \caption{All-black image size vs. runtime (remote).}
    \label{fig:nruntimeremote}
  \end{subfigure}
  \par 
\end{multicols}
\begin{multicols}{2}
  \begin{subfigure}[b]{\columnwidth}
    \includegraphics[width=\linewidth]{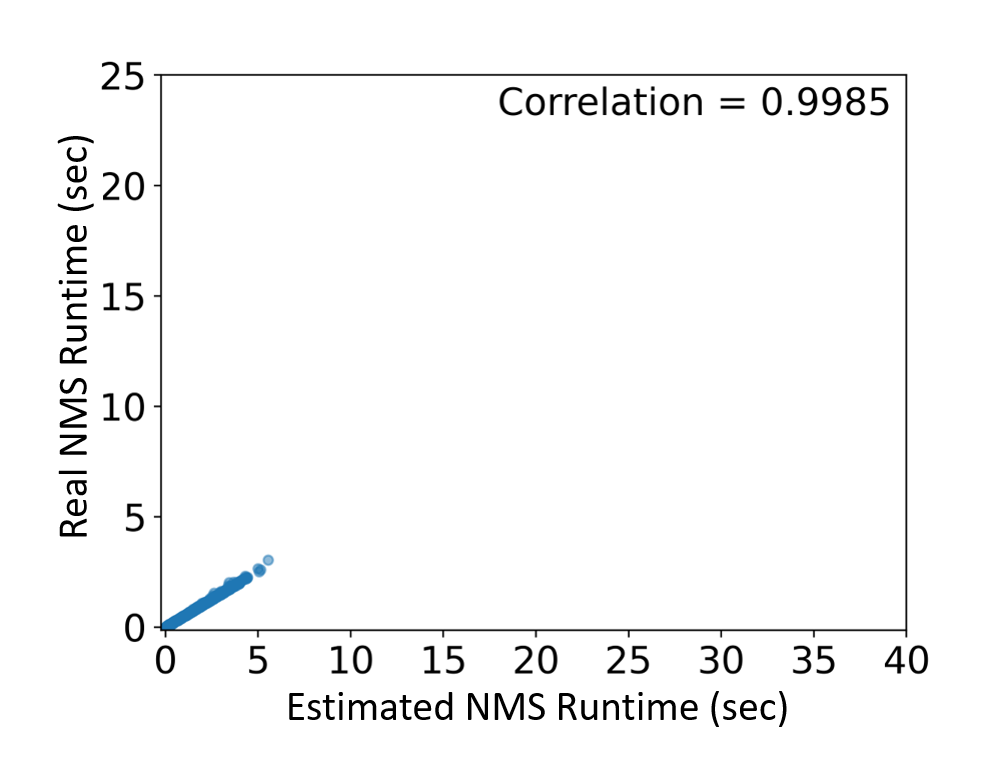}
    \caption{Real vs. estimated runtime (local).}
    \label{fig:estvsreallocal}
  \end{subfigure}
  \par 
  \begin{subfigure}[b]{\columnwidth}
    \includegraphics[width=\linewidth]{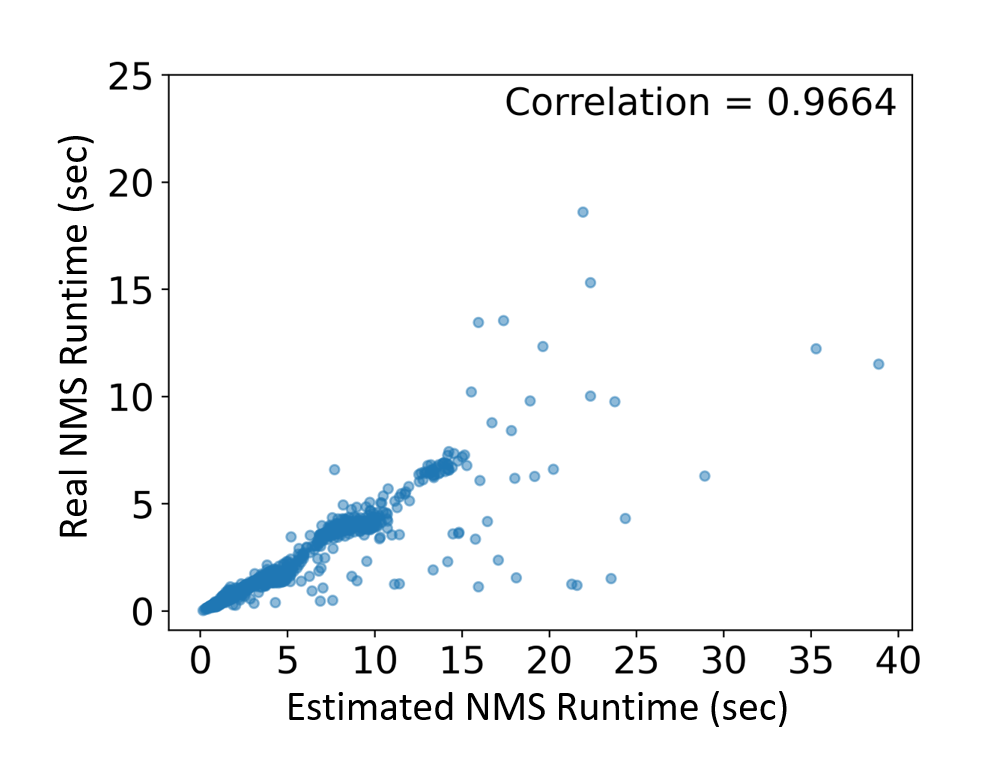}
    \caption{Real vs. estimated runtime (remote).}
    \label{fig:estvsrealremote}
  \end{subfigure}
\end{multicols}
\caption{Estimating the NMS runtime by modeling the neural component's runtime as a function of the image size, and subtracting it from the total runtime.}
\label{fig:bbs_nms_local_remote}
\end{figure}

\paragraph{Accurate remote measurement of leakage}
We repeat the experiment above, but this time in a VPN setup containing an HTTP client and server; both are implemented via Python scripts and run on an Intel core i7 CPU. 
The HTTP client repeatedly sends image POST requests to the HTTP server (which executes YOLOv3 on the given images) and measures the round-trip time (RTT) for each request (i.e., the time elapsed from the time the request was sent by the client until the server HTTP response arrives back to the client). 
Figs. \ref{fig:nruntimeremote} and \ref{fig:estvsrealremote} present the actual vs. estimated time, again showing close correspondence. 
Fig. \ref{fig:bbs_est_nms_remote} presents the correspondence between the remote attacker's estimated time and the number of bounding boxes; in this case, we also observe a high correlation, just slightly lower than those presented in Fig. \ref{fig:bbs_nms_local} where the measurements were obtained directly rather than estimated using a network connection.


\begin{figure}
\setlength{\columnsep}{0.5em} 
\begin{multicols}{2}
  \begin{subfigure}[b]{\columnwidth}
    \includegraphics[width=\linewidth]{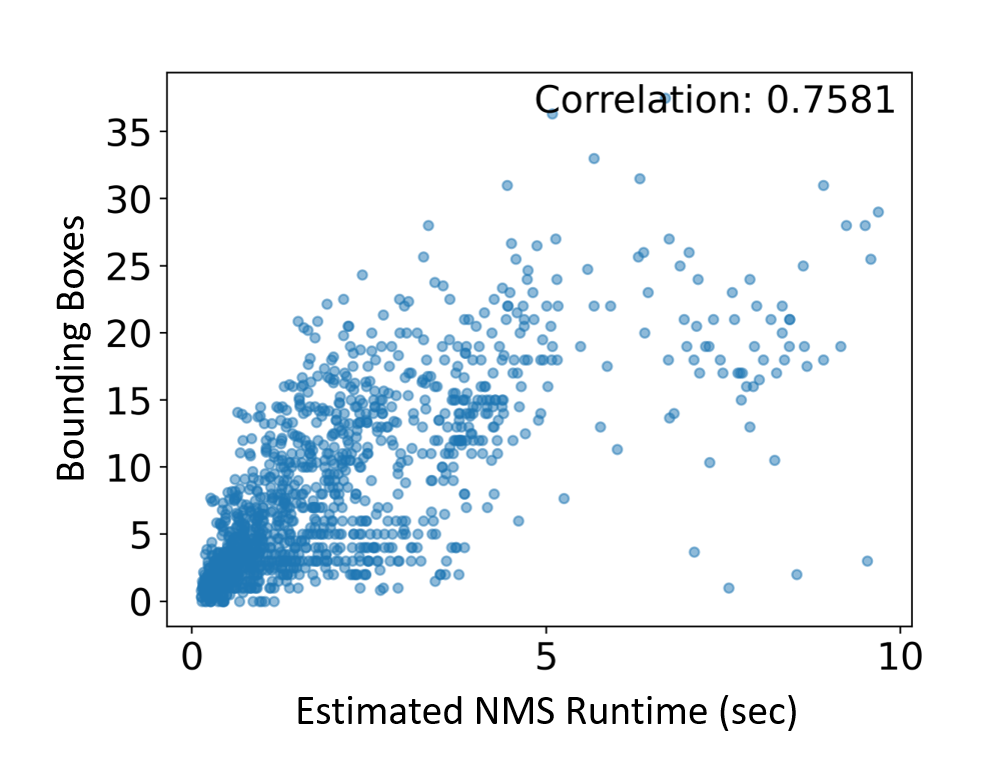}
    \caption{No amplification.}
    \label{fig:remotenoamp}
  \end{subfigure}
  \par 
  \begin{subfigure}[b]{\columnwidth}
    \includegraphics[width=\linewidth]{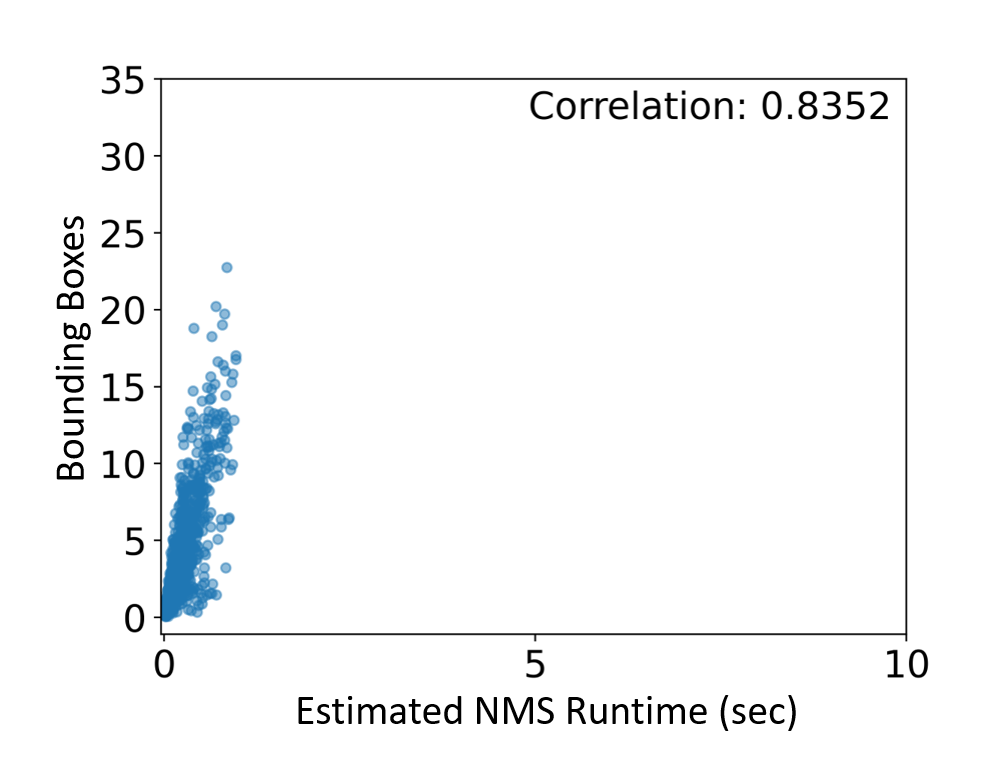}
    \caption{3x3 amplification.}
    \label{fig:remote33}
  \end{subfigure}
  \par 
\end{multicols}
\begin{multicols}{2}
  \begin{subfigure}[b]{\columnwidth}
    \includegraphics[width=\linewidth]{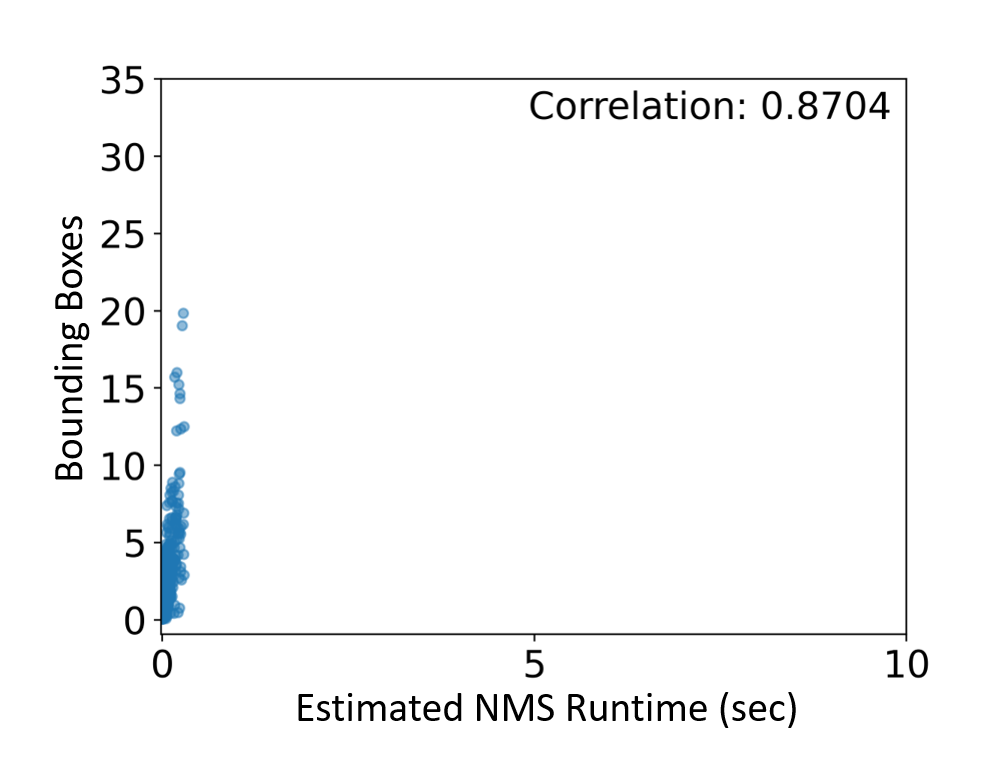}
    \caption{7x7 amplification.}
    \label{fig:remote99}
  \end{subfigure}
  \par 
\end{multicols}
\caption{\textbf{Estimated} NMS runtime vs. number of bounding boxes (\textbf{remote} querying).}
\label{fig:bbs_est_nms_remote}
\end{figure}
\section{Using Timing Leakage to Evade Detection}
\label{sec:evasion}

\begin{figure}
	\centering
	\includegraphics[width=0.5\textwidth]{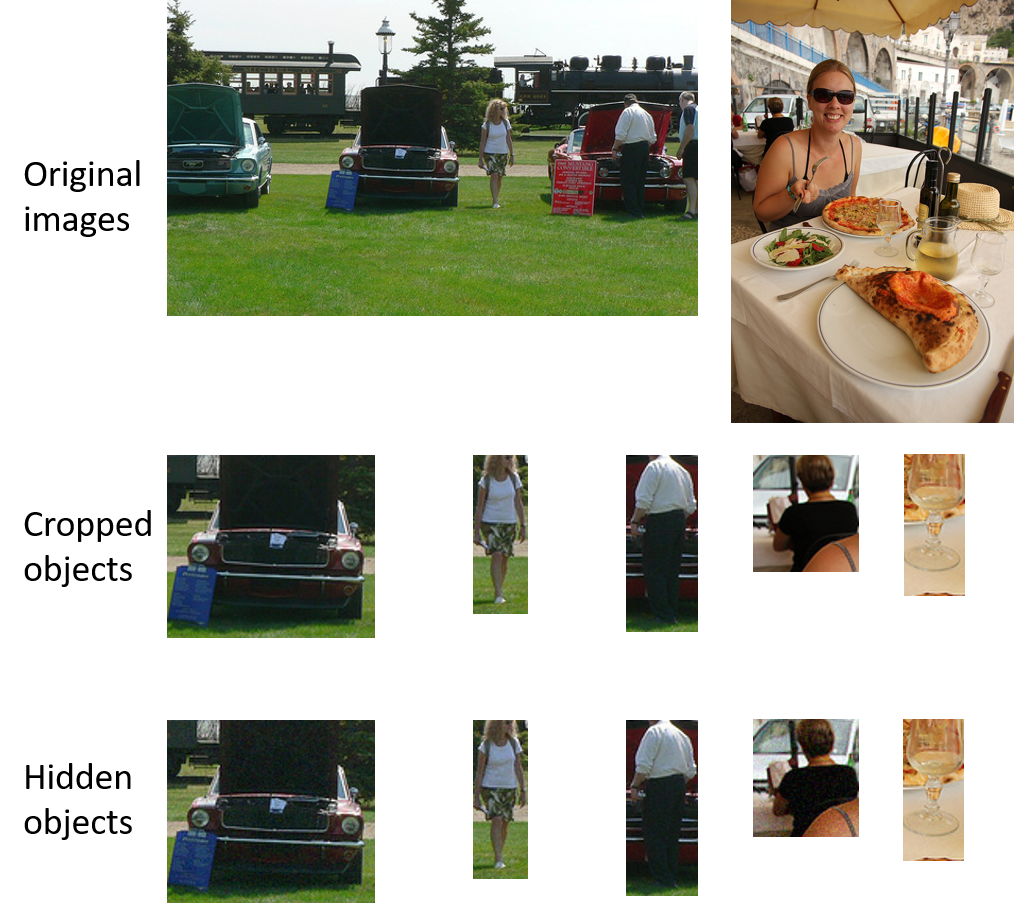}
	\caption{Examples of the timing-leakage-based evasion attack (original images, cropped objects, and hidden objects).}
	\label{fig:evasion-examples-1}
\end{figure}



\paragraph{Threat model}
Our adversary wishes to evade detection by applying adversarial perturbations on an image. We consider a black-box setting where the adversary can send an image to the object detector, receive a response, and measure the latency corresponding to the execution time of the inference procedure. In this scenario, the attacker has access to the output labels but does not have access to the detector's weights, architecture, or prediction confidence values.


\subsection{Timing-Leakage-Based Evasion Algorithm}

\paragraph{Exploring the use of timing as a proxy for confidence} Adversarial attacks that can access prediction confidence values usually try to gradually change them so that they will be inferior. 
For example, an evasion attack that accesses scores could work by iteratively trying to find a small perturbation that decreases the detector's confidence and applying it to the image, until no bounding boxes receive high enough confidence scores to be considered a detection. 
Our adversary has no access to the confidence values, so he/she uses timing as a proxy for confidence, leveraging the observation that they are connected, which was mentioned in Section~\ref{sec:leakage}.




\begin{algorithm}[h]
\caption{Timing-Leakage-Based Evasion Learning}
\label{alg:evade}
\begin{algorithmic}[1]
\State \textbf{Inputs:} 
\State \text{$gadget_0$ - the original cropped object}
\State \text{p - the size of the population}
\State \text{radius - controls population variance}
\State \text{$\lambda$ - determines the step size towards the new mutant} 
\State \textbf{Method:} 
\State $i \leftarrow 0$ 
\While {\text{not-detected?}($gadget_i$)}
	\State $amp_{i} \leftarrow amplify(gadget_{i})$
	\State $begin \leftarrow currentTimeInMills()$
	\State $result \leftarrow OD(amp_{i})$
	\State $end \leftarrow currentTimeInMills()$
	\State $execTime_{i} \leftarrow end-begin$
	\State \textbf{//** Creating a new population **//}
	\For{$j \leftarrow 0$ to p}
		\State \textbf{//** Draw a new instance in the population **//}
		\State $noise_{i,j} \leftarrow$ $ \text{draw uniformly from range [-1,1]}$
		\State $pert_{i,j} \leftarrow radius \times noise_{i,j}$
		\State $member_{i,j} \leftarrow pert_{i,j} + gadget_i$
		\State $amp_{i,j} \leftarrow amplify(member_{i,j})$
        \State $begin \leftarrow currentTimeInMills()$
        \State $result \leftarrow OD(amp_{i,j})$
        \State $end \leftarrow currentTimeInMills()$
        \State $execTime_{i,j} \leftarrow end-begin$
\EndFor	

\State \textbf{//** Calculating the fitness for each member**//}
\For{$j \leftarrow 0$ to p}
\State $fitness_{i,j} = \frac{|execTime_{i,j}-execTime_{i}|}{\sum_{m=1}^{p} |execTime_{i,m}-execTime_{i}|}$.
\State $direction_{i,j} \leftarrow sign (execTime_{i}-execTime_{i,j}) $
\EndFor	

\State \textbf{//** Calculating the mutation**//}
\State $mutation_{i} =\sum_{j=1}^{p} direction_{i,j} \times fitness_{i,j} \times pert_{i,j}$
\State $normMutation_{i} = \frac{mutation_{i}}{||mutation_{i}||_F}$

\State \textbf{//** Breeding:  calculating the mutation**//}
\State $gadget_{i+1} = gadget_{i} + (normMutation_{i} \times \lambda)$
\State $i \leftarrow i+1$ 
\EndWhile
\State \textbf{Output}: $gadget_i$
\end{algorithmic}
\end{algorithm}

 \begin{figure} 
	\centering
	\includegraphics[width=0.95\linewidth]{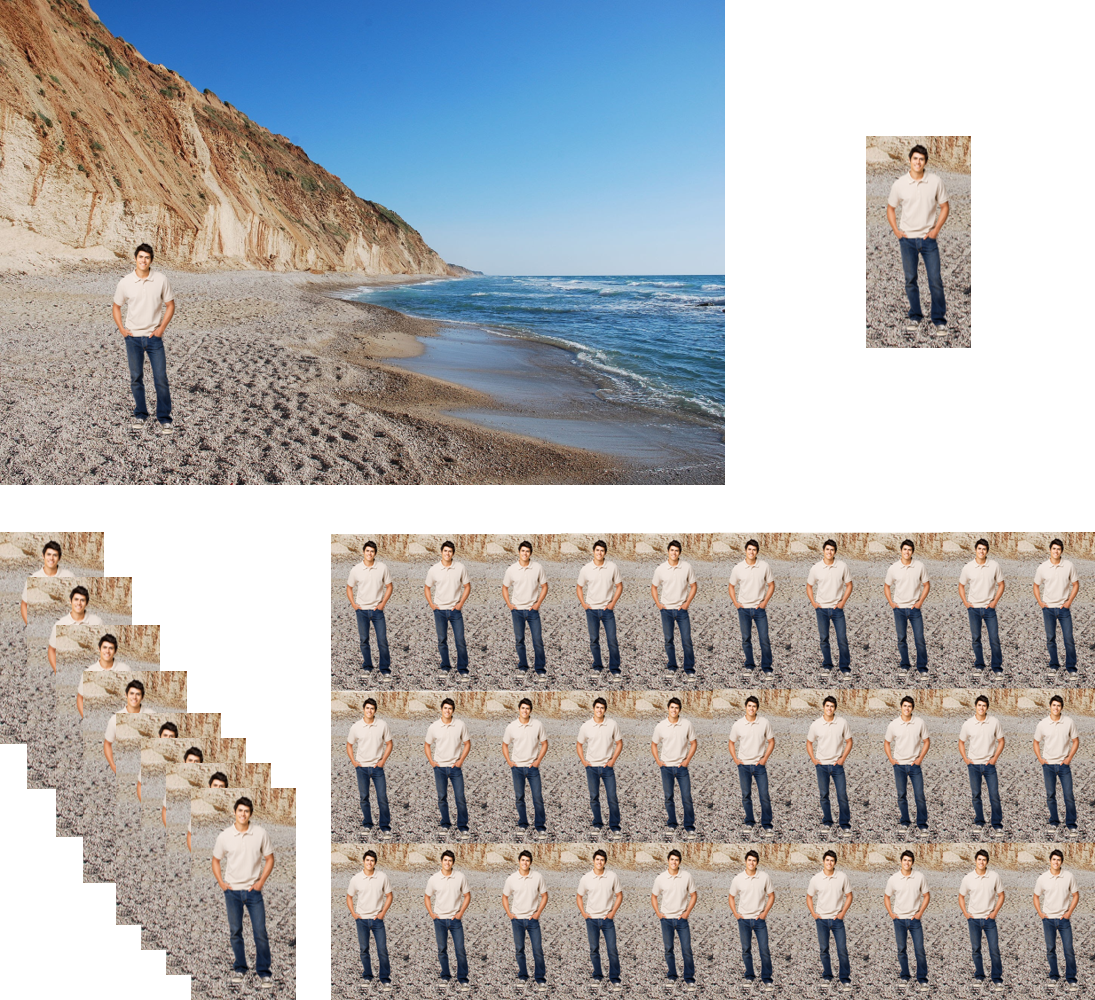}
	\caption{Timing-leakage-based evasion attack steps: Top left: the original image; top right: $gadget_0$ (the first gadget cropped from the image); bottom left: $member_{0,0}$,....,$members_{0,n}$ (the members are drawn, with the quantity depending on the population size.); bottom right: and $amplified_{0,0}$ (the amplified member created from the first member $member_{0,0}$).}
	\label{fig:method}
\end{figure}

\paragraph{Evolutionary algorithm}
Our attacker employs an iterative genetic algorithm, where in each iteration the algorithm draws instances (a population) near the object.
The instances are sent to the object detector for inference, and their execution times are used as fitness functions to approximate the quality of perturbations according to the principle that images with low execution times are more suggestive of evasion and are therefore more ``fit.''
The instances in the population are then \textit{bred} to create a new mutation. 
The steps of the attack are visualized in Fig. \ref{fig:method}.

Our attacker's input is the image of the object or the ``gadget,'' denoted as $gadget_0$. In each iteration $i$, the attacker performs the following steps: (1) Send $gadget_i$ to the object detector; return $gadget_i$ if it was \emph{not} detected. (2) Draw $n$ perturbations uniformly from a radius $radius$ around $gadget_i$ and apply them to the gadget to produce a new \textit{population}. (3) Send each population member to the detector and measure its runtime. (4) Calculate the member's \textit{fitness} using the measured runtime (see below). (5) Calculate the \textit{mutation} as the average of the member perturbations weighted by their fitness, multiplied by a learning-rate parameter $\lambda$. (6) Form $gadget_{i+1}$ by applying the perturbation mutation to $gadget_i$.

Whenever population members are sent to the detector, the attacker performs the leakage amplification described in Section~\ref{sec:leakage} on every member of the population by concatenating it multiple times. 
The fitness for population member $j$ in iteration $i$ is calculated as $fitness_{i,j} = \frac{execTime_{i}-execTime_{i,j}}{\sum_{m=1}^{n} |execTime_{i,m}-execTime_{i}|}$, where $execTime_{i}$ is the measured execution time of the gadget and $execTime_{i,j}$ is the measured runtime of population member $j$.
An implementation of this Algorithm is provided in Algorithm \ref{alg:evade}.

\subsection{Evaluation}

In this section, we evaluate the performance of the evasion attack.
The reader can use the following link\footref{fn:evasion-github} to download the code that implements the proposed method. We provide additional results related to the experiments done in this section in the appendix.

\paragraph{Experimental setup} The code was executed on a GPU cluster consisting of a few physical machines equipped with RTX 2080 Ti, six cores, and 32 GB RAM.
To adhere to the detector's size requirements, we resized the images to 416x416 pixels before sending them to YOLOv3.

The reader can assess the quality of the adversarial instances visually by looking at the original and hidden objects in Fig. \ref{fig:evasion-examples-1} and below quantitatively based on $L_2$ norm.
In the rest of this section, we refer to the $L_2$ norm as the evasion budget needed to hide an object from YOLO using the algorithm.

\begin{figure}
\centering
\includegraphics[width=0.38\textwidth]{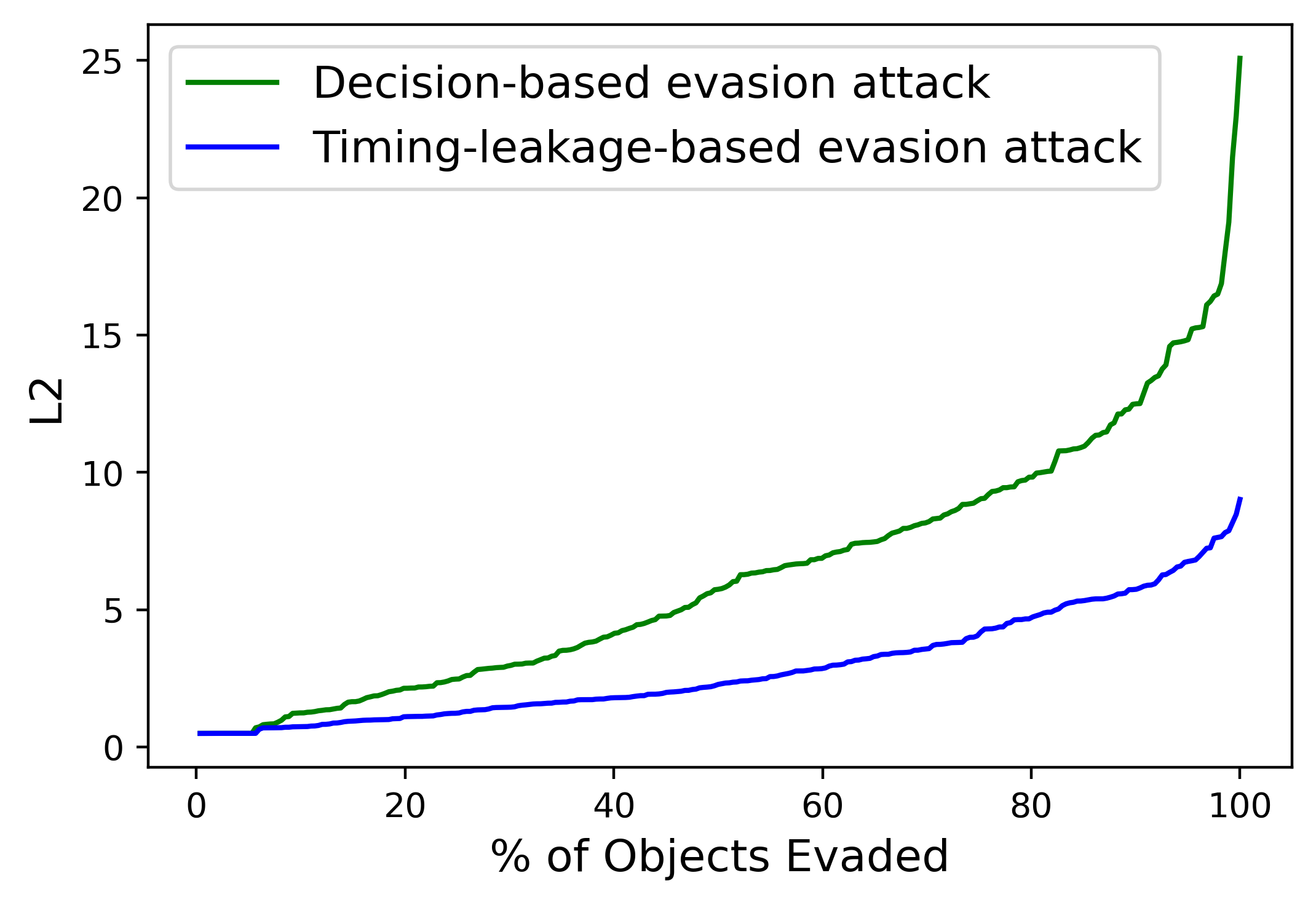}
\caption{The evasion budget ($L_2$ norm) required to evade a percentage of objects in the dataset when timing leakage was taken into account and when it was not considered (baseline). The results were calculated based on normalized RGB values (0-1).} \label{fig:eval-1}
\end{figure}

\paragraph{Comparison to a decision-based attack} We now compare the performance of the evasion algorithm to a modified version of the algorithm which does not take the timing leakage into consideration.
The following values were used in our code: $p = 20$ (the size of the population drawn in each iteration), $\lambda= 0.5$ (determines the size of the step towards the new mutant), and $radius = 25.0$ (the difference between the gadget and the drawn instances of the population).
The modified decision-based evasion attack bases its decision solely on the output of the object detector. 
To produce the modified version, we changed the fitness calculation to set the fitness of each population member to $\frac{1}{n}$ if the object was not detected or $-\frac{1}{n}$ if it was detected and omitted the leakage amplification (since the baseline does not use leakage). 
The remaining values ($\lambda$, $radius$, and $p$) are equal in both versions.
We randomly selected 280 images from the COCO dataset for this experiment.

We performed the timing-leakage-based and decision-based evasion attacks and computed the $L_2$ norm perturbation produced by each of the attacks.
The results are presented in Fig. \ref{fig:eval-1}, where it can be seen that the perturbation size is significantly smaller for the timing-based evasion attack than it is for the decision-based evasion attack, demonstrating the former's superiority over the latter.


\begin{figure}
\centering
\includegraphics[width=0.38\textwidth]{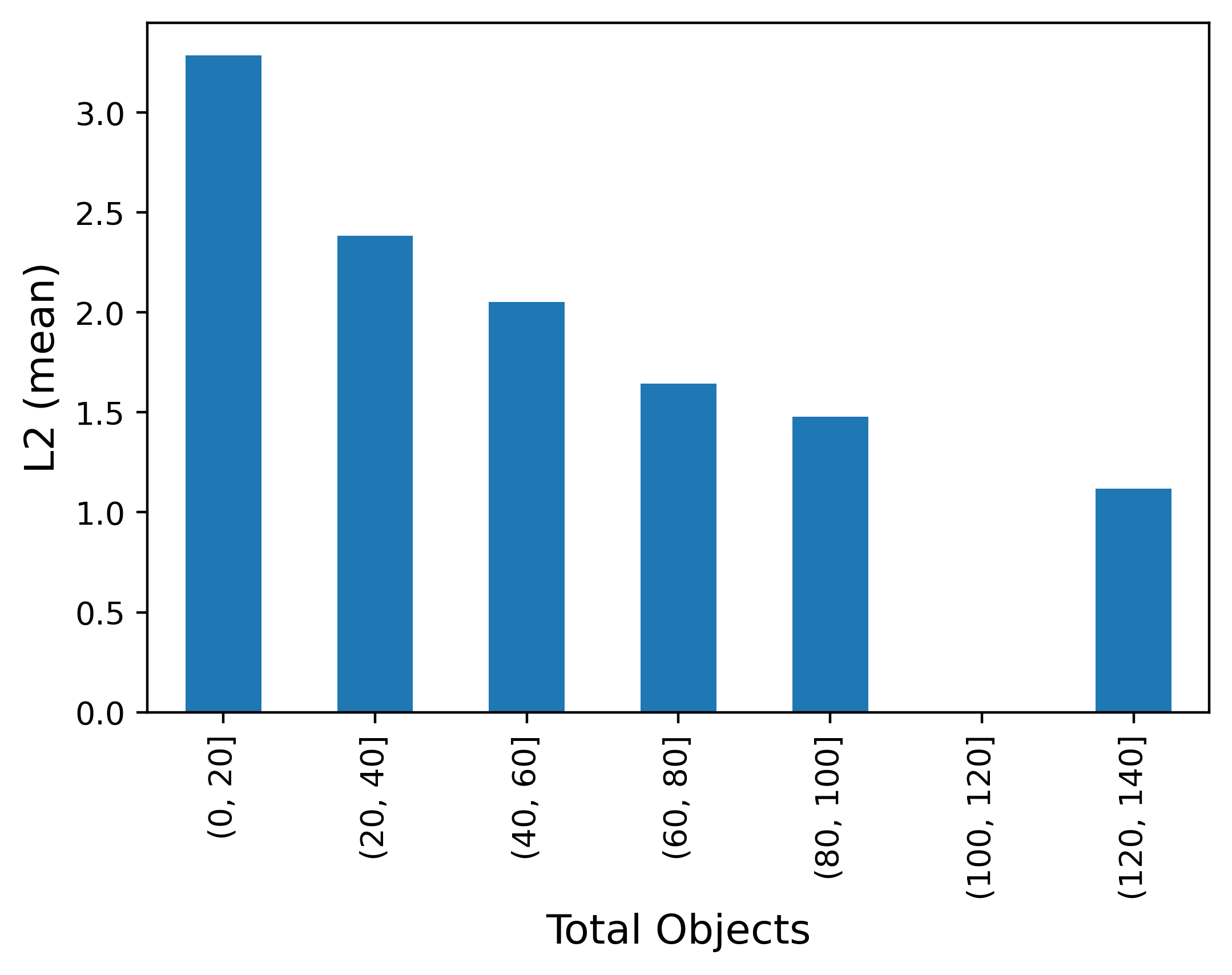}
\caption{The effect of the number of objects on the initial amplification: the average evasion budget ($L_2$ norm) needed to cause an object to evade an object detector as a function of the number of objects detected with $amplification_{0}$. The results were calculated based on normalized RGB values (0-1).} \label{fig:eval-2}
\end{figure}

\paragraph{The effect of the number of objects initially detected in the amplified image} 
Some objects that appear in the original gadget do not appear in all of its concatenated copies in the leakage amplified version (likely due to the resize operation, which introduces a resolution difference between them). 
We now examine the implications of this effect on our attack.

For each of the 280 images in this experiment, we measured the number of objects detected in its leakage amplified version (prior to any perturbation). 
Fig. \ref{fig:eval-2} presents the correspondence between this number and the attack's success in terms of the perturbation $L_2$ norm. 
The results show that the more ``successful'' the amplification procedure is in creating detected copies of the object (and thus in amplifying the leakage), the better our attack performs, further confirming the effect of leakage on performance and indicating that the attack can be improved by amplifying it further.

\begin{figure}
\centering
\includegraphics[width=0.38\textwidth]{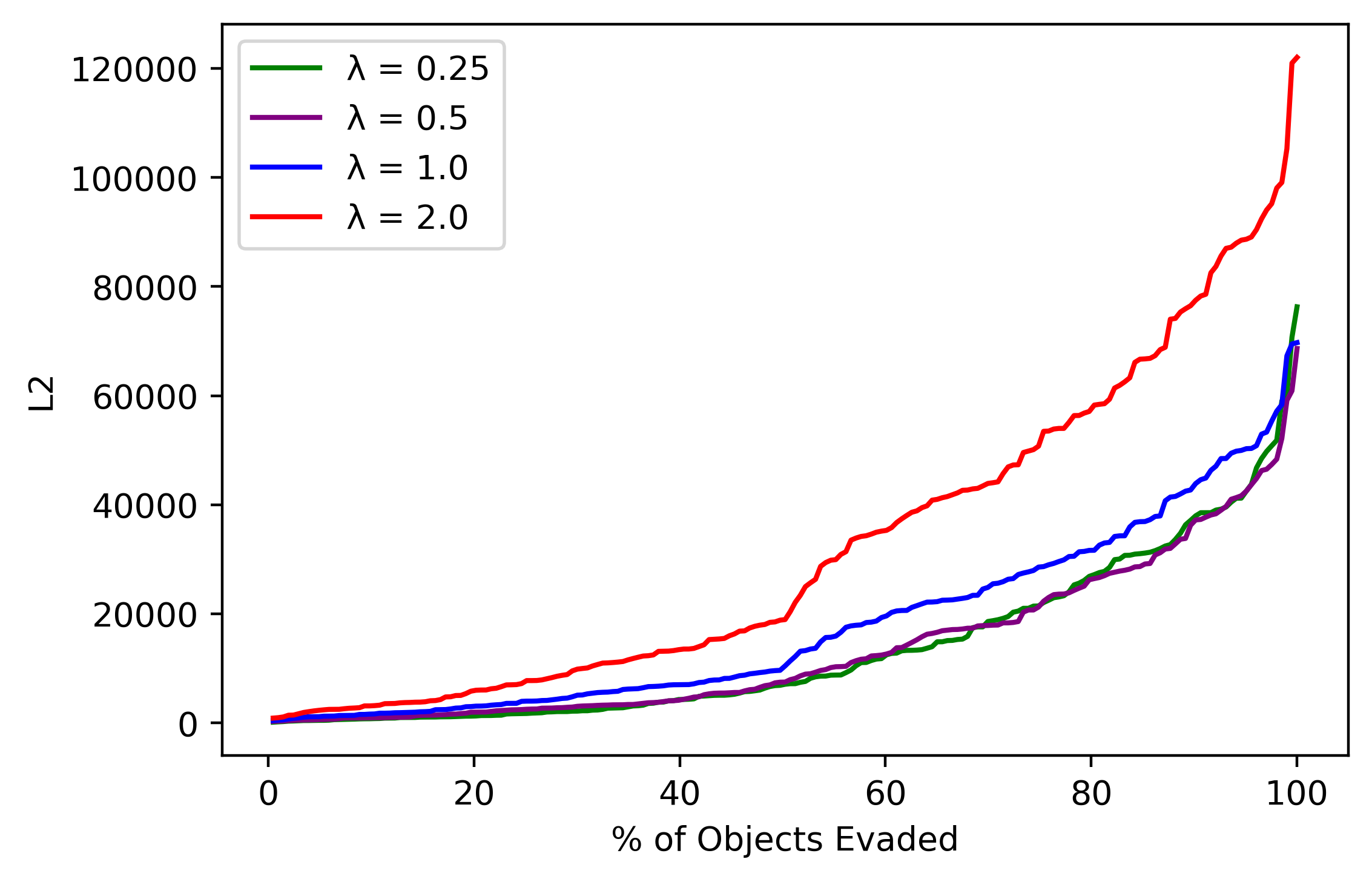}
\caption{The average evasion budget ($L_2$ norm) needed to evade a percentage of objects in the dataset with varying values of the $\lambda$ parameter. The results were calculated based on RGB values (0-255). The lower the value of $\lambda$, the smaller the evasion budget.} \label{fig:eval-budget}
\end{figure}

\paragraph{The effect of the $\lambda$ value (which determines the size of the step towards the new mutant) on the evasion budget}
We now evaluate the performance of the evasion attack when it is used with different $\lambda$ values.
We randomly selected 100 images from the COCO dataset for this set of experiments.
We ran our code in four experiments, and in each experiment, we changed the value of $\lambda$ (0.25, 0.5, 1.0, 2.0).
The rest of the parameters were fixed as follows in the experiments.
We used $p = 20$ and $radius = 25.0$. 
We computed the mean $L_2$ norm on the original image and the hidden image for each of the objects and computed the mean $L_2$ norm for each $\lambda$ value.
The results are presented in Fig. \ref{fig:eval-budget}.
Unsurprisingly, the evasion budget decreases when the value of $\lambda$ decreases.




\section{Dataset Inference Using Timing Leakage}
\label{sec:membership-inference}

\newcommand{\norm}[1]{\left\lVert#1\right\rVert}
\newcommand{\brac}[1]{\left(#1\right)}
\newcommand{\curbrac}[1]{\left\{#1\right\}}
\newcommand{\sqbrac}[1]{\left[#1\right]}
\newcommand{\abrac}[1]{\langle#1\rangle}
\newcommand{\size}[1]{\left|#1\right|}

\newcommand{\memberset}{\mathcal{D}_m}
\newcommand{\nonmemberset}{\mathcal{D}_{\widetilde{m}}}
\newcommand{\targetset}{\mathcal{T}}
\newcommand{\memberavg}{\hat{\mu}_m}
\newcommand{\nonmemberavg}{\hat{\mu}_{\widetilde{m}}}
\newcommand{\targetavg}{\hat{\mu}_{\mathcal{T}}}
\newcommand{\membermean}{\mu_{m}}
\newcommand{\nonmembermean}{\mu_{\widetilde{m}}}

\paragraph{Threat model}
Given a \textit{target set} of examples and black-box query access to an object detector, our attacker tries to determine whether or not the entire set of examples was part of the detector's training set. 
Our attacker does not necessarily aim to identify the members of a single image, but rather the inclusion of a specific set or a source of images (for example, the set of images including a specific individual). 
Conservatively, we assume that the set examples are drawn independently and identically distributed from the same dataset used to produce the attacker's training data (otherwise, we can expect the attacks to be easier, as the attacker's signal could potentially reveal instance membership as well as a distributional similarity to the training data).

We assume the adversary already has two sets of labeled data points with samples known to be members of the victim's training set and with known nonmembers.\footnote{This is often a required assumption for dataset inference; a large amount of research has shown that the attacker can somewhat mitigate the need for this by profiling membership vs nonmembership behavior on shadow models constructed in an offline phase
However, this imposes another requirement, since it assumes that the attacker can draw examples from a distribution similar to the victim model's training set. See~\cite{shokri2017membership} and follow-up studies in Section~\ref{sec:related-work}.}

\paragraph{Attack method}
First, the attacker characterizes the runtime behavior of the model. This can be done once for use in inferring membership on multiple target sets. This includes the following steps: (1) use the method described in Section~\ref{sec:leakage} to learn how the neural component's runtime correlates to the image size; this allows the attacker to approximate the NMS algorithm's runtime by subtracting the neural component's runtime from the total runtime; and (2) query the model and record the estimated NMS runtime. Prior to querying the model, our attacker first amplifies the expected timing leakage by tiling or concatenating each image multiple times (see Section~\ref{sec:leakage}).

Then, given a target set, the attacker queries the model with every instance in the set, again after amplification, and records the runtime.

The attacker now possesses a sample of member runtimes, a sample of nonmember runtimes, and a sample of target set runtimes, which are denoted respectively as $\memberset, \nonmemberset$, and $\targetset$. The attacker approximates the expected proportion of samples with a runtime \(\geq75\) seconds that $\memberset$ and $\nonmemberset$ are sampled from by calculating the average, denoted as $\memberavg$ and $\nonmemberavg$. Next, the attacker examines whether the average $\targetavg$ is closer to $\memberavg$ or $\nonmemberavg$; if the average $\targetavg$ is closer to $\memberavg$, the attacker determines that the target set contains training set members, and if it is closer to $\nonmemberavg$, it contains nonmembers.

\paragraph{Experimental setup}
We extracted 2,000 RGB images containing various objects (i.e., people, animals, and vehicles) from seven YouTube video recordings of street views taken by people traveling in different physical locations (i.e., cities): New York (NY, USA), San Francisco (CA, USA), Dubai (United Arab Emirates), Miami (FL, USA), London (UK), Los Angeles (CA, USA), and George Town (Singapore). Due to dynamic movements, for each object, the image was taken at a slightly different angle with respect to the object and a different distance from the object.


We trained YOLO to detect ``person,'' the most common label in the COCO dataset, on 2,000 randomly selected images in the COCO dataset's training set that contains a person. We used the default hyperparameters from the official YOLO code~\cite{darknetyolo}, such as a batch size of 64 and training for 6,000 iterations. Early stopping was used to choose the weights with the highest mAP (mean average precision) on the validation set (500 images). An mAP of 48.70\% was achieved, which is not far from YOLO's reported 55.3\%.



We executed YOLO on a machine with a Titan Xp GPU and we performed inference on 2,000 randomly chosen images that were used for training and 2,000 randomly chosen images from the COCO dataset's training set that were not used for training, and timed each run using Python's \texttt{time} module. Before performing inference on an image,  we amplified the image by tiling it in a 5x5 pattern. We also resized the image height and width to the closest multiple of 32, to comply with YOLO's input size requirements.
As in Section~\ref{sec:leakage}, we timed YOLO's run on all-black images of varying sizes to learn how the neural component's runtime corresponds to the various image sizes and subtracted the estimated neural runtime from the total runtime, producing estimates for the instances' NMS runtimes.

\paragraph{Calculating the attacker's likelihood (i.e., probability) of success}
Our attack only needs to be performed once but to empirically measure the attacker's likelihood (i.e., probability) of success, we would have to simulate the attack many times, each time sending multiple queries. A far more scalable way to understand the risk posed by this attack is to estimate the distribution of the attacker's timing samples and reason analytically about their success.

The previous step resulted in a set of runtimes for a training set of members and nonmembers.
We computed a histogram of values for each set which shows that there is a notable drop in the frequency of samples with an estimated NMS runtime over 75 seconds. Using this value as a benchmark, we define the following: 

\[X_{i} =
    \begin{cases}
      0 & \text{if NMS runtime < 75 seconds}\\
      1 & \text{if NMS runtime $\geq$ 75 seconds}
    \end{cases} \]

The sum of these indicator variables is defined as \(X = \sum_{i=1}^{n}X_{i} \). We use the member and nonmember runtime data to determine the expected value \(\mathbb{E}(X)\) and use this variable to calculate bounds on the attacker's false positive rate. We outline our method here and provide the full details in the appendix.

Let the mean of \(X\) from above be $\mu_{m}$ for the member set and $\mu_{\bar{m}}$ for the nonmember set. Let $h=\size{\mu_{m}-\mu_{\bar{m}}}/4$. Let $\memberavg$ be the average of $\memberset$ and $\nonmemberavg$ be the average of $\nonmemberset$. Assume that the attacker's target set $\targetset$ contains only nonmembers of the training set, and let $\targetavg$ be its average. To bound the attacker's false positive rate, we use Chernoff's bound to upper bound the probabilities for the following events: (1) $\size{\memberavg-\membermean}>h$, (2) $\size{\nonmemberavg-\nonmembermean}>h$, and (3) $\size{\targetavg-\nonmembermean}>h$. We then apply a union bound on the three events. Based on this, we observe that if all three events do not hold, then by triangle inequality we know that $\size{\targetavg-\membermean}<\size{\targetavg-\nonmembermean}$ and our attacker will succeed. Therefore, we can bound the attacker's false negative rate using a symmetrical argument.

{}


\begin{figure}
\captionsetup{aboveskip=0pt}
\setlength{\columnsep}{1em} 
\begin{multicols}{2}
  \begin{subfigure}[b]{\columnwidth}
    \includegraphics[width=\linewidth, height=4cm]{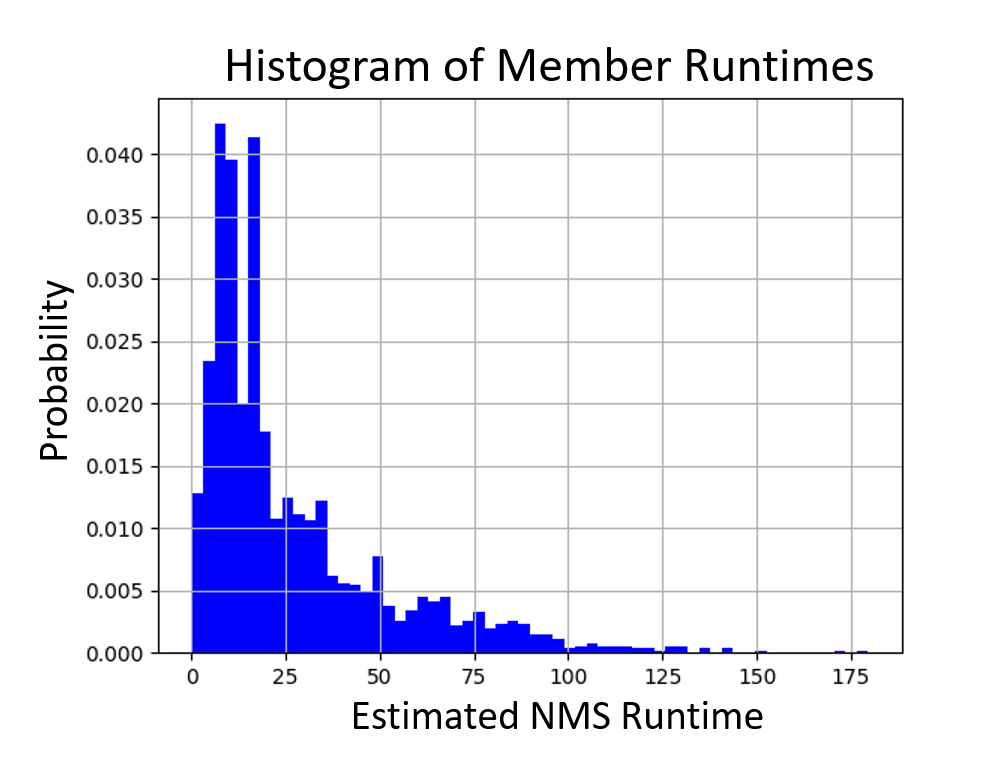}
    \label{subfig:member-histogram}
  \end{subfigure}
  \par 

  \begin{subfigure}[b]{\columnwidth}
    \includegraphics[width=\linewidth, height=4cm]{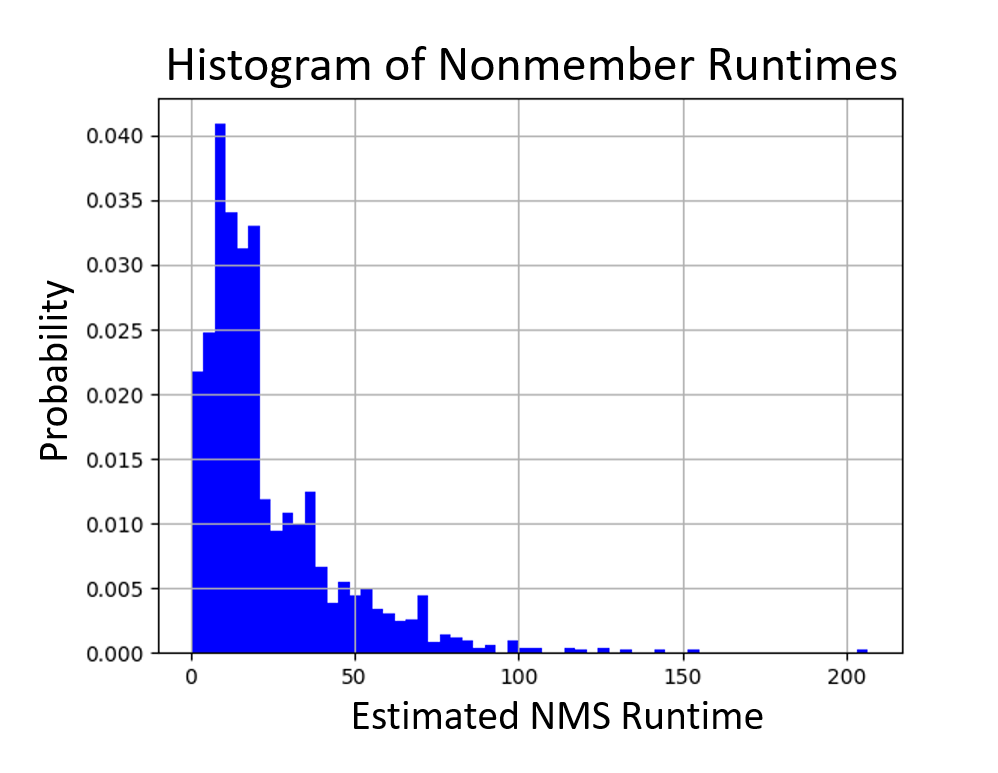}
    \label{subfig:nonmember-histogram}
  \end{subfigure}
  \par 

\end{multicols}
\caption{Timing histograms for member (left) and nonmember (right) samples.}
\label{fig:membership-details}
\end{figure}

\begin{figure}
\centering
\includegraphics[width=0.38\textwidth]{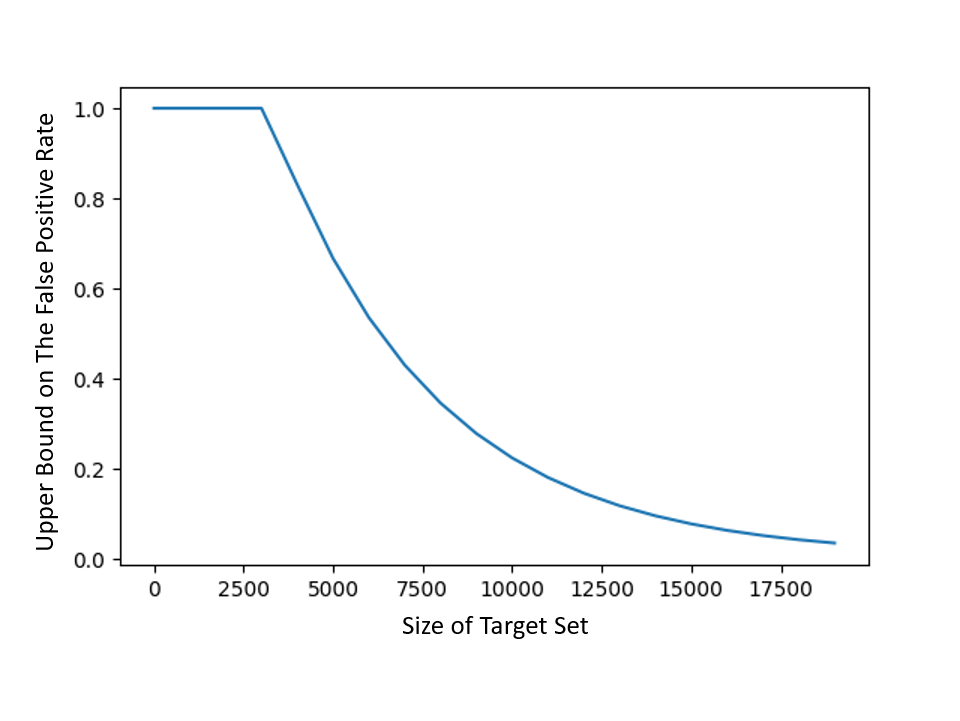}
\caption{Upper bound on the false-positive rate as a function of the attacker's sample size $\size{\targetset}$.} \label{fig:success}
\end{figure}

\paragraph{Results and analysis}
YOLO detected the objects in about 75\% of the images.
Our distributions have means of 6.8\% and 2.9\% for images with runtimes \(\geq75\) seconds, respectively for the member and nonmember sets. Fig. \ref{fig:membership-details} presents the timing histograms for members and nonmembers. These means are multiplied by the assumed member and nonmember set sizes to determine \(\mathbf{E}[X]\). Fig. \ref{fig:success} shows the bound on our false positive rate, which exponentially decreases with the size of the attacker's target set.

\section{Related Work}
\label{sec:related-work}

\paragraph{Adversarial inputs and dataset inference}
Adversarial inputs such as the ones we construct to evade detection are an extensively studied subject, and many such attacks have been proposed~\cite{Carlini-Wagner, FGSM, BIM, ilyas2018black, ilyas2018prior, chen2017zoo, chen2020hopskipjumpattack, brendel2017decision}; most pertinent to our work are black-box attacks, and specifically ``decision-only'' attacks~\cite{chen2020hopskipjumpattack, brendel2017decision} that only have the capability of observing model decision outputs. Attacks against object detectors have also been demonstrated~\cite{nassi2020phantom, song2018physical, liu2018dpatch, zhao2019seeing, wang2021daedalus}, including Daedalus~\cite{wang2021daedalus} which specifically targets the NMS component of object detectors, fooling it to produce many false detections. Prior black-box attacks against object detectors are transfer-based~\cite{wang2021daedalus} and have a disadvantage in that they can be effectively mitigated using ensemble adversarial training~\cite{tramer2017ensemble}. Our work is the first to mount a decision-based attack on an object detector, and we improve it using side-channel leakage.



Shokri et al.~\cite{7958568}, whose work was followed by a long line of studies~\cite{song2019membership,choquette2021label,shafran2021membership}, exploited the target model's score to infer training set membership. We are the first to show dataset inference in a label-less setting where the model's decision is not needed.


\paragraph{Side-channel attacks}
Side-channel attacks exploit unintended and externally measurable side effects of information processing, such as program runtime or shared-resource contention, to extract sensitive information such as cryptographic keys \cite{kocher1996timing,brumley2005remote} or website-visit identity~\cite{panchenko2016website,schuster2017beauty}.
Prior side-channel attacks on neural networks primarily focused on extracting network architecture and weights from various side channels including cache, hardware-component timing, or electromagnetic emanations \cite{duddu2018stealing,batina2019csi,yoshida2019model, hua2018reverse}. Sun et al.~\cite{sun2020anonymizing} fingerprinted inputs via a cache attack. Nakai et al.~\cite{nakai2021timing} used timing leakage from embedded microcontrollers to guide the search for adversarial perturbations, however the proposed attack is only applicable to specific hardware and requires physical access to the embedded device.

\textit{Algorithmic runtime side channels.} An especially powerful class of side channels, algorithmic runtime side channels leak sensitive information through the variable time of algorithmic components. They are characterized by strong leakage signals, allowing attackers to exploit them remotely using noisy measurements with very limited knowledge of the underlying implementations~\cite{brumley2005remote, schwarzl2021practical, gluck2013breach}. We are the first to use timing as a side channel to target the integrity and privacy of machine learning.




\section{Countermeasures}
\label{sec:countermeasures}

Countermeasures should address the trade-off between usability and security, decoupling the runtime of inference from sensitive information, by running in constant time (necessarily worst case). 
For example, while incorporating random delays might be a viable strategy against limited attacker models, it does not always prevent leakage~\cite{van2014encyclopedia} and ultimately faces a similar limitation. 
To effectively conceal secret-dependent leakage, the delays would generally need to be on the order of magnitude of the runtime variance caused by the leakage itself, which can be considerable.

In the case of the NMS algorithm, we note that striking this balance is highly nontrivial. 
Greedy NMS can, in the worst case, run for many minutes or even hours on common setups, as demonstrated in this paper. 
Employing a constant-time greedy NMS approach would result in even slower processing times, rendering it impractical for real-time systems. 
Thus, the development of countermeasures must carefully navigate the trade-off between usability and security, aiming to effectively mitigate runtime leakage while maintaining practicality in real-world applications.

As for more efficient approaches than greedy NMS, Neubeck et al.~\cite{neubeck2006efficient} studied NMS variants and observed a trade-off between efficiency and ease of implementation. 
The variants studied are not constant-time, but they would potentially allow mitigation via delays. 
Libraries and frameworks that implement efficient and parallelizable versions of the NMS algorithm would be a significant step forward, as they would allow practitioners to avoid the highly-leaky greedy variant.

Another potential avenue is neural or learned approaches for NMS~\cite{hosang2017learning} or end-to-end object detectors that do not use NMS at all. 
Recently, attention-based end-to-end approaches achieved performance that is on par with NMS-based ones~\cite{carion2020end,zhou2019objects,sun2021sparse}. 
Our attack highlights the substantial advantages, in terms of security and privacy, of these approaches.
\section{Limitations}
\label{sec:limitations}

Our concrete attack models of detection evasion and dataset inference make some conventional assumptions, such as the ability to query the detector multiple times (for evasion) and knowledge of portions of the training data (for dataset inference). These assumptions are described and justified earlier in the paper.

It is also important to note the specificity of the attacks presented in this paper, which are only applicable to object detectors that incorporate the NMS algorithm. In Section \ref{sec:countermeasures}, we emphasize the significance of exploring alternative approaches to NMS in object detectors. Additionally, in Section \ref{sec:discussion}, we discuss potential avenues for future research that involve extending these attacks to encompass a broader range of architectures.
\section{Discussion \& Future Work}
\label{sec:discussion}

\paragraph{Conclusion and broader impact} This study demonstrates the potential advantage attackers can gain from variable-time inference algorithms, especially those that involve non-neural components if they are not implemented and deployed with caution. 
The observed leakage signal-to-noise ratio in our case study on the NMS algorithm is remarkably strong and easily measurable, even over an Internet connection. 
Furthermore, the unique characteristics of object detection operations allow attackers to significantly amplify their signal-to-noise ratio. 
We show that this leakage can be exploited in attacks that outperform the decision-only baseline. 
These findings highlight the power of algorithmic timing side channels, which we are the first to study and utilize for attacks on machine learning, and we raise awareness of the risks posed to machine learning models and the need for constant-time mitigation.

\paragraph{Future work}  Future work should expand the scope to include hybrid architectures like multi-exit networks or generative language models.
In our future investigations, we also plan to delve further into the possibilities of leveraging temporal information leakage in the NMS algorithm. This exploration will lay the foundation for developing attacks that are solely based on timing aspects. It is important to emphasize that these extensions of the attack would be applicable in a threat model in which the prediction and bounding box returned from inference are not available or present. Our plan to focus on this aspect reflects the significance of studying the timing side channel in scenarios where access to certain outputs may be restricted. 
These extensions would involve enhancing the attack presented in Section \ref{sec:evasion} to solely exploit timing vulnerabilities.

In future research, we plan to explore analogous attacks exploiting NMS leakage (Section \ref{sec:leakage}), as well as attacks described in Sections \ref{sec:evasion} and \ref{sec:membership-inference} on different hardware components. We will also investigate the influence of network traffic on these attacks.

\section*{Acknowledgements}

This work was partially supported by the Cyber Security Research Center at Ben-Gurion University of the Negev, the Jacobs Urban Tech Hub at Cornell Tech, and the Technion's Viterbi Fellowship for Nurturing Future Faculty Members.
We would like to acknowledge our sponsors, who support our research with financial and in-kind contributions: Amazon, Apple, CIFAR through the Canada CIFAR AI Chair, DARPA through the GARD project, Intel, Meta, and the Sloan Foundation. Resources used in preparing this research were provided, in part, by the Province of Ontario, the Government of Canada through CIFAR, and companies sponsoring the Vector Institute. We would also like to thank CleverHans lab group members for their feedback.

\bibliographystyle{IEEEtran}
\bibliography{IEEEabrv,main}
\section{Appendix}
\label{sec:appendix}

\subsection{Evasion Attack Supplementary Material} 
The material below supplements the information presented in Section~\ref{sec:evasion}.


\begin{figure}[h]
  \begin{subfigure}[b]{1.0\columnwidth}
    \includegraphics[width=\linewidth]{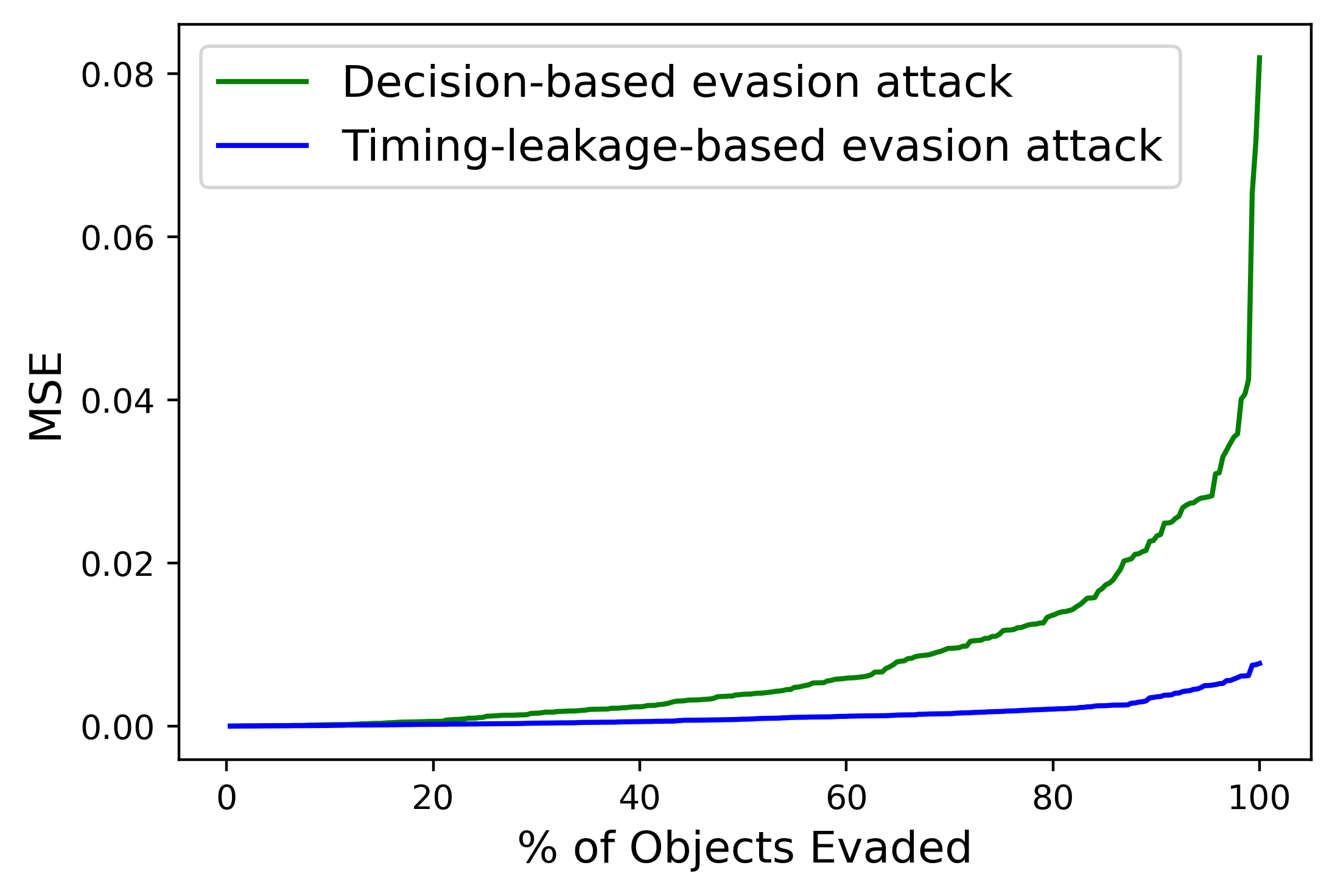}
  \end{subfigure}
  \hfill 
  \begin{subfigure}[b]{1.0\columnwidth}
    \includegraphics[width=\linewidth]{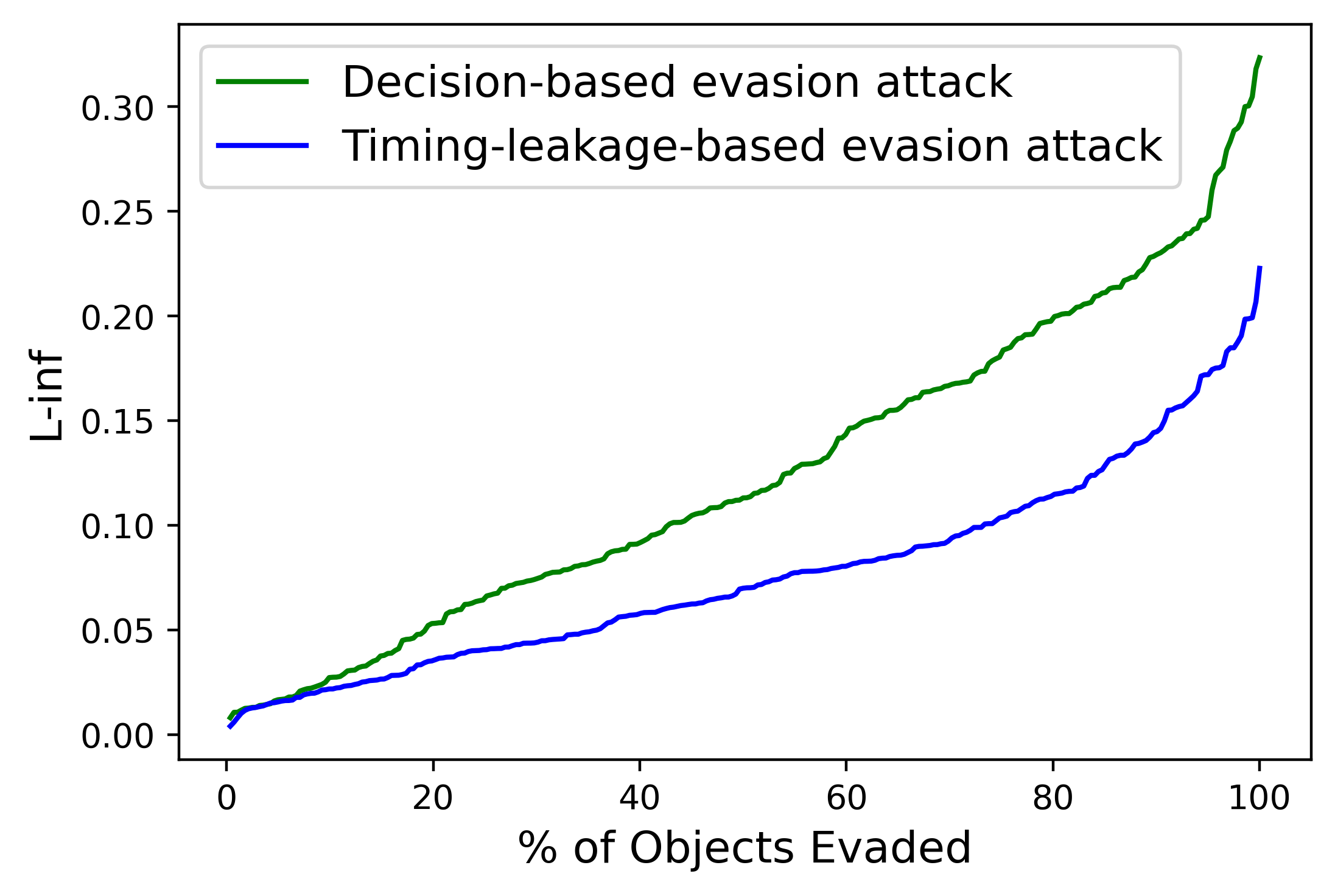}
  \end{subfigure}
\caption{Evasion budget for timing and baseline attacks: evading a percentage of objects in the dataset with and without considering timing leakage (normalized RGB values 0-1).}
\label{fig:baseline}
\end{figure}

Fig. \ref{fig:baseline} presents the results for additional metrics (MSE and $L_{\infty}$) used to compare the performance of our timing-leakage-based evasion attack and decision-based evasion attack (described in Section~\ref{sec:evasion}).


\begin{figure}[h]
  \begin{subfigure}[b]{1.0\columnwidth}
    \includegraphics[width=\linewidth]{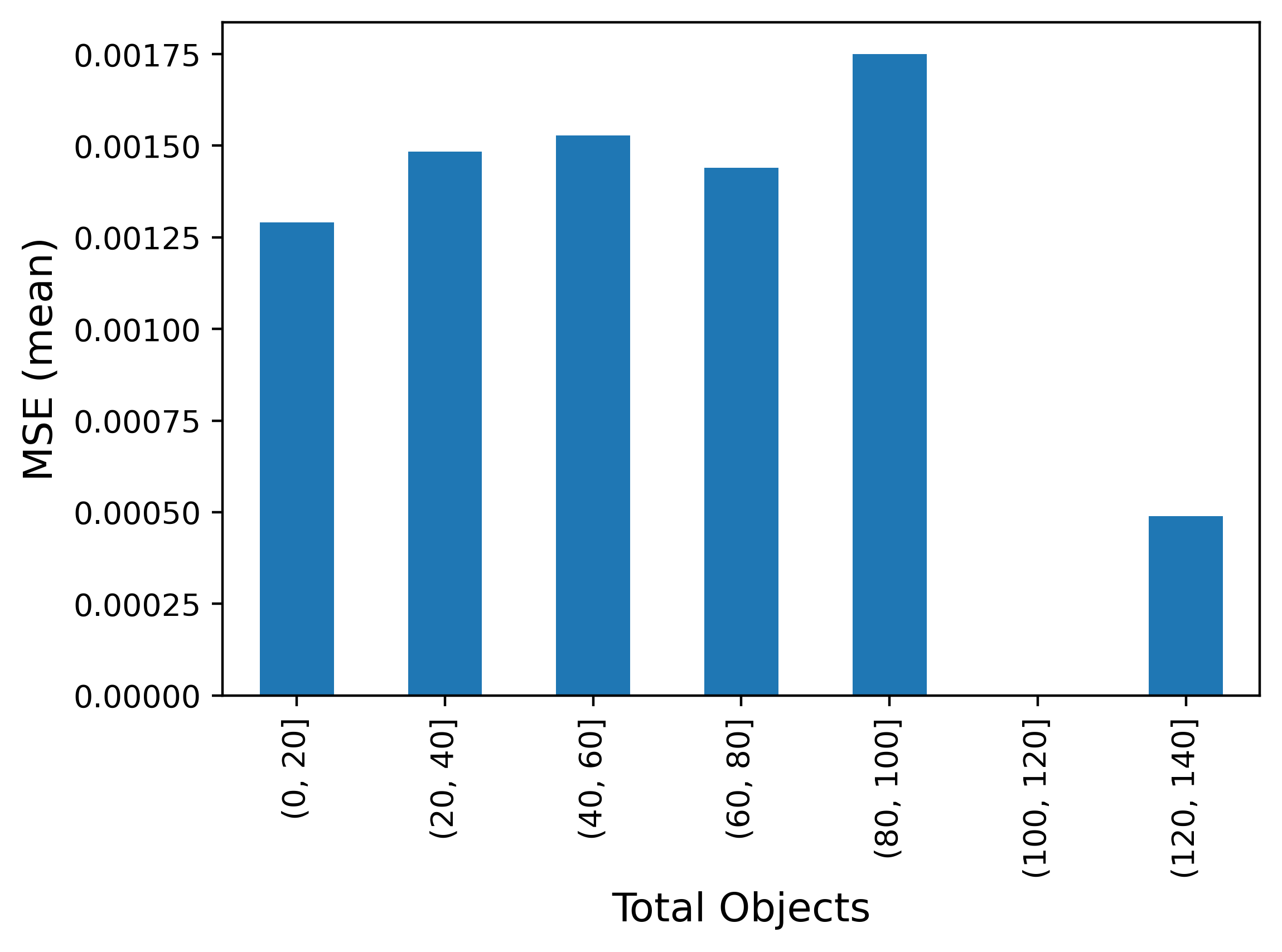}
  \end{subfigure}
  \hfill 
  \begin{subfigure}[b]{1.0\columnwidth}
    \includegraphics[width=\linewidth]{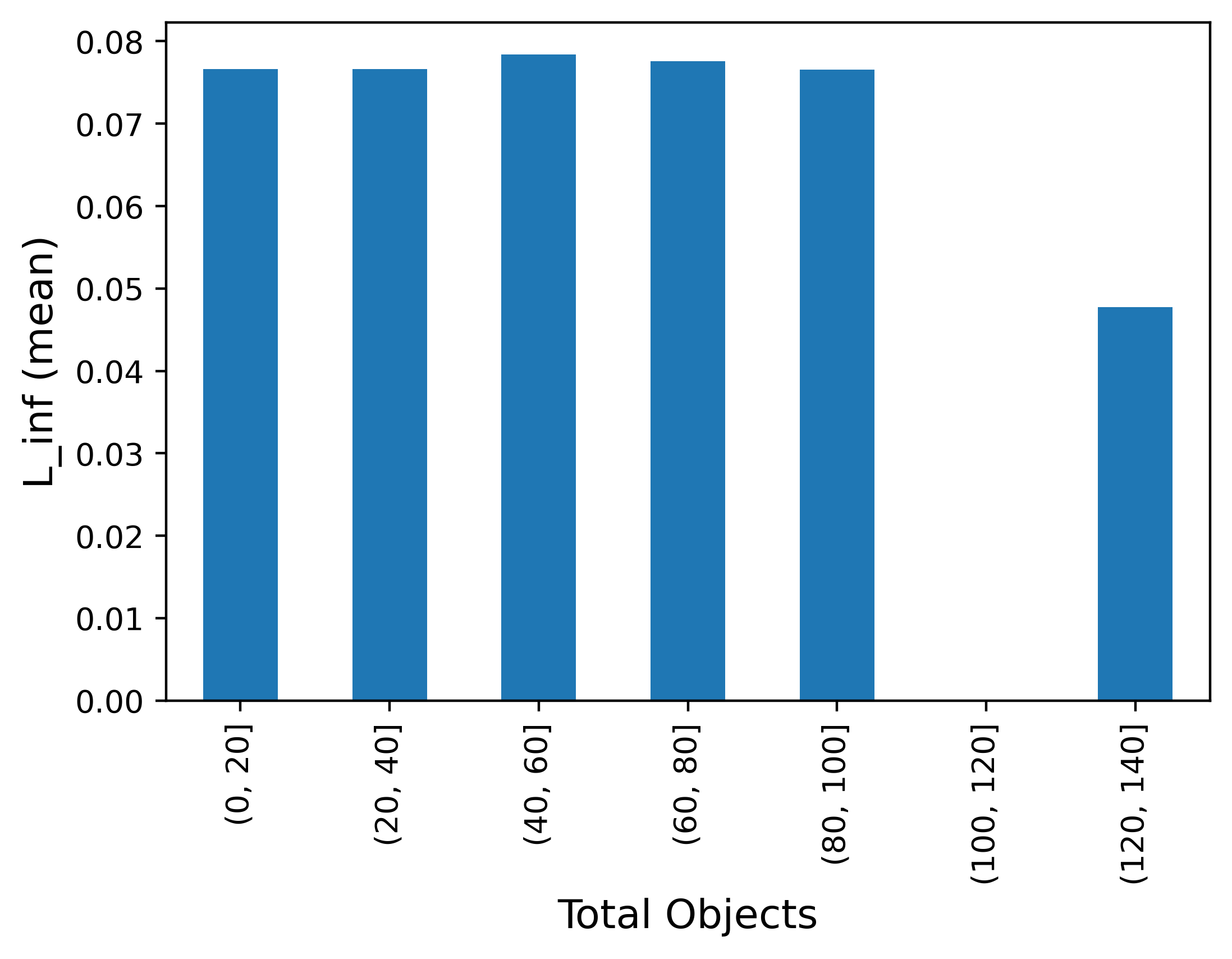}
  \end{subfigure}
\caption{The effect of the number of objects on the initial amplification: the average evasion budget as a function of the number of objects detected with $amplification_{0}$ (normalized RGB values 0-1).}
\label{fig:number}
\end{figure}

Fig. \ref{fig:number} presents the results for additional metrics (MSE and $L_{\infty}$) used to demonstrate the effect of the number of objects detected with $amplification_{0} $ (discussed further in Section~\ref{sec:evasion}).



\subsection{Dataset Inference Attack: Analysis Details}
The material below supplements the information presented in Section~\ref{sec:membership-inference}.

\paragraph{Chernoff bound} 
Let $X$ be the sum of $n$ random i.i.d. indicator variables, and let $\mu = \mathbf{E}[X]$. For any $0\leq\delta\leq 1,$
\begin{equation*}
    \mathbf{Pr}[|X-\mu|\geq\delta\mu] \leq 2e^{{-\mu\delta^{2}/3}}
\end{equation*}

For example, to bound the probability of the event $\size{\memberavg-\membermean}\leq h$, we set $\mu\gets n\membermean$, $\delta\gets h/\membermean$, where n is the assumed size of the member set.
\paragraph{Union-bounding failure probability} Our analysis relies on the following claim.
\begin{theorem}
If (1) $\size{\memberavg-\membermean}\leq h$ and (2) $\size{\nonmemberavg-\nonmembermean}\leq h$ and (3) $\size{\targetavg-\nonmembermean}\leq h$, then
$\size{\targetavg-\nonmemberavg}\leq\size{\targetavg-\memberavg}$.
\end{theorem}
\begin{proof}\renewcommand{\qedsymbol}{}

From (2), (3), and triangle inequality, we know that $\size{\nonmemberavg-\targetavg}\leq 2h$. Assume the claim is incorrect, implying that $\size{\memberavg-\targetavg}<2h$. Then we obtain $\size{\membermean-\nonmembermean}=\size{\membermean-\targetavg+\targetavg-\nonmembermean}\leq\size{\membermean-\targetavg}+\size{\targetavg-\nonmembermean}\leq\size{\membermean-\memberavg}+\size{\memberavg-\targetavg}+\size{\targetavg-\nonmembermean}<4h$, which contradicts how $h$ is chosen.

\end{proof}

\footnotesize 
\Urlmuskip=0mu plus 1mu\relax

\end{document}